\documentclass[10pt]{article}

\usepackage{amsmath, amsthm}
\usepackage{amssymb}

\usepackage{hyperref}

\usepackage{graphicx}

\usepackage{cite}

\usepackage[usenames,dvipsnames]{xcolor} 
\usepackage{multirow}
\usepackage{array}
\usepackage{tikz}
\usepackage{mathtools}
 \usepackage{float}

\usepackage{ctable}
\usepackage{multirow} 
\usepackage{mathtools}
 
 \newtheorem{lemma}{Lemma}
\newtheorem{theorem}{Theorem}
 
\usepackage[labelfont=bf,labelsep=period,justification=raggedright]{caption}

\makeatletter \renewcommand{\@biblabel}[1]{\quad#1.} \makeatother

\newcommand{\lra}[1]{\langle #1 \rangle}

\newcommand{\turnover}{\nu}

\newcommand{\fosp}{s}
\newcommand{\sampleprop}{sampling proportion}

\newcommand{\leafp}{s_l}
\newcommand{\leafsampleprop}{leaf sampling proportion}

\newcommand{\intNum}{l}

\date{}

\begin{document}

\begin{flushleft} {\Large \textbf{Bayesian inference of sampled ancestor trees
for epidemiology and fossil calibration}
}
\\ Alexandra Gavryushkina$^{1,2, \ast}$, David Welch$^{1,2}$, Tanja Stadler$^3$,
Alexei Drummond$^{1,2, \ast}$ \\ \bf{1} Department of Computer Science,
University of Auckland, Auckland, New Zealand \\ \bf{2} Allan Wilson Centre for
Molecular Ecology and Evolution, New Zealand \\ \bf{3} ETH Z\"{u}rich,
Switzerland \\ $\ast$ E-mail: Corresponding sasha.gavryushkina@auckland.ac.nz,
alexei@cs.auckland.ac.nz \end{flushleft}

\section*{Abstract}

Phylogenetic analyses which include fossils or molecular sequences that are
sampled through time require models that allow one sample to be a direct
ancestor of another sample. As previously available  phylogenetic inference
tools assume that all samples are tips, they do not allow for this possibility.

We have developed and implemented a Bayesian Markov Chain Monte Carlo (MCMC)
algorithm to infer what we call sampled ancestor trees, that is, trees in which
sampled individuals can be direct ancestors of other sampled individuals.

We use a family of birth-death models where individuals may remain in the tree
process after the sampling, in particular we extend the birth-death skyline
model [Stadler {\it et al}, 2013] to sampled ancestor trees. This method allows
the detection of sampled ancestors as well as estimation of the probability that
an individual will be removed from the process when it is sampled. We show that
sampled ancestor birth-death models where all samples come from different time points are non-identifiable and thus
require one parameter to be known in order to infer other parameters.

We apply this method to epidemiological data, where the possibility of  sampled
ancestors enables us to  identify individuals that infected other individuals after being sampled
and to infer fundamental epidemiological parameters.

We also apply the method to infer divergence times and diversification rates
when fossils are included among the species samples, so that fossilisation events are
modelled as a part of the tree branching process. Such modelling has many
advantages as argued in literature.

The sampler is available as an open-source BEAST2 package \\
(\url{https://github.com/gavryushkina/sampled-ancestors}).

\section*{Author Summary}

A central goal of phylogenetic analysis is to estimate evolutionary relationships and
population parameters such as speciation and extinction rates or the 
rate of infectious disease spread from molecular data. 
The statistical methods used in these analyses require that  the underlying tree branching process be
specified. Standard models for the branching process which were
originally designed to describe the evolutionary past of present day species do not
allow for direct ancestors within the sampled taxa, that is, they do not allow one sampled taxon to be the 
ancestor of another. However the probability of
sampling a direct ancestor is not negligible for many types
of data. For example, when fossil and living species are analysed together to
infer species divergence times, fossil species may or may not be direct
ancestors of living species. In epidemiology, a sampled individual (a host from
which a pathogen sequence was obtained) can infect other individuals after sampling, which then go on to be sampled themselves. 
Recently, models that allow for direct ancestors have been
introduced. Such models produce phylogenetic trees with a different structure from the classic phylogenetic trees
and so using these models in inference requires new computational methods.
Here we developed and implemented a Bayesian Markov chain Monte Carlo 
framework for phylogenetic analysis allowing for the possibility of direct ancestors.

\section*{Introduction}

Phylogenetic analysis uses  molecular sequence data   to infer
evolutionary relationships between organisms and to infer evolutionary
parameters. Since the introduction of Bayesian inference in phylogenetics
\cite{Yang1997,Mau1999,MrBayes}, it has become the standard approach for fully
probabilistic inference of evolutionary history with many popular
implementations~\cite{Beast17,MrBayes32,Beast2,PhyloBayes3} of Markov chain
Monte Carlo (MCMC)~\cite{Metropolis1953,Hastings1970} sampling over the space of
phylogenetic trees. Initial descriptions of Bayesian phylogenetic analysis were
restricted to considering bifurcating trees \cite{Yang1997,Mau1999}, but have been extended to 
include explicit polytomies~\cite{Lewis2005}. Here we tackle phylogenetic
inference with trees that may contain sampled ancestors \cite{Gavr2013}.

Standard phylogenetic models developed for inferring the evolutionary past of
present day species assume that all samples are terminal (leaf) nodes
in the estimated phylogenetic tree. However, serially sampled data
generated by different evolutionary processes can be analysed using phylogenetic
methods \cite{Drummond2003} and, in some cases, the assumption that all sampled taxa are leaf
nodes is not appropriate.

One case in point is when inferring epidemiological
parameters from viral sequence data obtained from infected
hosts~\cite{Pybus2001,Gren2004,Stadler2011R0,Ypma13,Kuhnert2014}. Viral sequences are obtained from distinct hosts and treated as samples
from the transmission process. Using standard models to describe the infectious
disease transmission process entails the assumption that a host becomes
uninfectious at sampling (where sampling is obtaining a sequence or
sequences from the pathogen population residing in a single infected host).
However in many cases, hosts remain infectious after sampling and, when sampling is sufficiently 
dense, the probability of sampling an individual that later infects 
another individual which is also sampled  is not negligible
\cite{volz2013inferring, Teunis13, Vrancken2014}. 

A recent analysis of a well-characterised HIV transmission chain
\cite{Vrancken2014} employed a hierarchical model of a gene tree inside a
transmission tree to infer the differences in evolutionary rates (substitution
rates) within and among hosts. Hierarchical modelling of gene trees inside
transmission trees has also been used to infer transmission events for small
epidemic outbreaks where epidemiological data is available in the form of known
infection and recovery times for each host~\cite{Ypma13}. In both cases the
inference of transmission trees assumes complete sampling of the hosts involved,
and the host sampling process is not explicitly modelled.

Incomplete sampling is explicitly modelled by birth-death-sampling
models~\cite{Stadler2010JTB,Stadler2011R0, Didier12, StadKuhn12}, for which the
probability density functions of the trees are available in closed form, thus making these models
tractable for use in Bayesian inference. The birth-death-sampling models do not
assume that individuals are removed from the tree process after the sampling. 
However, applications of models allowing for infection after sampling has not been possible due to a lack of software, meaning that many applications simply ignore sampled
ancestors~\cite{Stadler2011R0, StadKuhn12}.

Another problem that may require sampled ancestor models
 is inferring  species divergence times using fossil data. 
Without the means to calibrate the times of divergences, the length of branches in
the estimated molecular phylogeny of contemporaneous sequences are typically
described in units of expected substitutions per site. Geologically
dated fossil data can be employed to calibrate a phylogenetic tree, thus
providing absolute branch lengths. The most common approach here is to specify
age limits or a probability density function on specific divergence times in the
phylogeny, where the constraints are defined using the fossil data
\cite{Sanderson1997,Thorne1998,Drummond2006,Rannala2007,Ho2009}. There
are several drawbacks connected to this approach~\cite{Ronq2012,Heath2013}.
First, there is  potential for  inconsistency when applying two priors on
the phylogeny \cite{Heled2012}: a calibration prior on one or more divergence
times and a tree process prior on the entire tree. Second, it is not obvious how
to specify a calibration density so that it reflects  prior knowledge about
 divergence times \cite{Ronq2012,Heath2013}. Finally, such densities usually
only use the oldest fossil within a particular clade, thus discarding much of
the information available in the fossil record \cite{Heath2013}.

An approach that addresses  these issues requires modelling fossilisation
events as a part of the tree process prior. This allows for the joint analysis
of fossil and recent taxa together in a unified framework~\cite{Pyron2011,
Wood2012, Ronq2012,Heath2013,Silv2014}.  Models that jointly describe the
processes of macroevolution and fossilisation should account for possible
ancestor-descendant relationships between fossil and living
species~\cite{Foote1996}, and thus include sampled ancestors.

A birth-death model with sampled ancestors have been used to estimate 
speciation and extinction rates from phylogenies in~\cite{Didier12}. 
Heath {et al.}~\cite{Heath2013} have used the birth-death model with sampled
ancestors (they call this the {\it fossilized birth-death process}) to
explicitly model fossilisation events and estimate divergence times in a
Bayesian framework.  In their approach, the tree topology relating the
extant species has to be known for the inference \cite{Heath2013}. So a method
that simultaneously estimates the divergence times and tree topology while
modelling incorporation of sampled fossil taxa is an obvious next step.

Full Bayesian MCMC inference using models with sampled ancestors is
complicated by the fact that such models produce trees, which we call  {\it sampled ancestor trees}~\cite{Gavr2013},  
 that are not strictly binary. They may have sampled nodes that 
lie on branches, forming an internal node with one direct ancestor and one
direct descendent.   Thus, modelling sampled
ancestors induces a tree space where the tree has a variable number of
dimensions (a function of the number of sampled ancestors), which necessitates extensions to the standard MCMC tree algorithms.  

Here we
describe a reversible-jump MCMC proposal kernel~\cite{Green1995} to effectively traverse the
space of sampled ancestor trees and implement it within the BEAST2 software
platform~\cite{Beast2}. We  study the limitations of  birth-death models
with sampled ancestors and extend the birth-death skyline model~\cite{StadKuhn12} to sampled
ancestor trees. We apply the new posterior sampler to two types of data: a
serially sampled viral data set (from HIV), and molecular phylogeny of bear
sequences with fossil samples.

\section*{Methods}

\subsection*{Tree models with sampled ancestors}

In this section, we consider birth-death sampling models \cite{Stadler2010JTB,
Stadler2011R0, Didier12, StadKuhn12} under the assumption that sampled individuals are not
necessarily removed from the process at sampling. This results in a type of
phylogenetic tree that may contain degree two nodes called \emph{sampled
ancestors}. 

An important characteristic of the models we consider here is incomplete
sampling, i.e., we only observe a part of the tree produced by a process.
Consider a birth-death process that starts at some point in time (the time of
origin) with one lineage and then each existing lineage may bifurcate or go
extinct. Further, the lineages are randomly sampled through time. An example of
\emph{a full tree} produced by such process is shown in
Figure~\ref{fig:fullsampledtree} on the left. We have information  only about
the portion of the process that produces the samples, shown as labeled nodes,
and do not observe the full tree. Thus we only consider this subtree relating to
 the sample, which is called the \emph{reconstructed tree} (or the \emph{sampled
tree}) and is shown on the right of Figure~\ref{fig:fullsampledtree}.

\subsubsection*{The sampled ancestor birth-death model}

Here we describe a serially-sampled birth-death model with sampled ancestors
\cite{Stadler2010JTB, Stadler2011R0}.

The process begins at the time of origin $t_{or} >  0$ measured in time units before the present. 
Moving towards the present, 
each existing lineage bifurcates or goes extinct according to two
independent Poisson processes with constant rates $\lambda$ and  $\mu$,
respectively. Concurrently, each lineage is sampled with Poisson rate $\psi$ and
 is removed from the process at sampling with probability $r$. The process is stopped at time $0$. This process can
be used to model the transmission of infectious disease and we call it \emph{the transmission birth-death process}.

The transmission process produces trees that have degree two nodes corresponding
to sampling events when a lineage was sampled but was not removed. We call these
trees \emph{sampled ancestor trees} (whether or not any sampled ancestors are
present).  The reconstructed tree  has degree-two nodes when a lineage is
sampled but not removed and then it, or a descendent lineage, is sampled again. 
The reconstructed tree in Figure~\ref{fig:fullsampledtree} (on the right) is an example 
of a sampled ancestor tree. Note that the root of a sampled ancestor tree
is the most recent common ancestor of the sampled nodes and therefore it may be a sampled 
node. There is no origin node in the tree because the 
time of origin is a model parameter and not an outcome of the process. 

A tree (or genealogy) $g$ consists of the discrete component $\mathcal T$,
which is called \emph{a tree topology}, and the continuous component $\bar
\tau$, which is called \emph{a time vector}. The tree topology of a sampled
ancestor tree is \emph{a sampled ancestor phylogenetic tree}, which is a ranked
labeled phylogenetic tree with labeled degree-two vertices (a rigorous
definition of a sampled ancestor phylogenetic tree can be found
in~\cite{Gavr2013}, where it is called an FRS tree). The time vector is a
real-valued vector of the same dimension as the number of ranks (nodes) in the
tree topology and with coordinates going in the descending order so that each
node in the tree topology can be unambiguously assigned a time from the time
vector. 

Further, we have three types of nodes: bifurcation nodes, sampled tip nodes,
sampled internal nodes. Let $m$ be the number of leaves, then $m-1$ is the
number of bifurcation events. Let $\bar x = (x_1, \ldots, x_{m-1})$ be a vector
of bifurcation times, where $x_{m-1} <\ldots < x_1$. Let $\bar  y = (y_1,
\ldots, y_m)$ be a vector of tip times, where $y_m < \ldots < y_1$. Further let
$\bar  z = (z_1, \ldots, z_k)$ be a vector of times of sampled two degree nodes,
where $z_k < \ldots < z_1$ and $k$ is the number of such nodes. Then $\bar \tau$
can be obtained by combining elements of $\bar x$, $\bar y$, and $\bar z$ and
ordering them in the descending order (see also
Figure~\ref{fig:fullsampledtree}). A genealogy may be written as $(\mathcal
T, \bar x, \bar y, \bar z)$. 

Stadler et al. \cite{Stadler2011R0} derive the density  of a genealogy $g =
(\mathcal T, \bar x, \bar y, \bar z)$ given the transmission birth-death process parameters $\lambda,\mu,\psi,r$ and time of origin
$t_{or}$. In \cite{Stadler2010JTB}, it was indicated that we should also condition on the event, $S$, of sampling at least one individual
because only non-empty samples are observed. 
The density is

\begin{equation}\label{epid} f[g | \lambda, \mu, \psi, r, t_{or}, S] = \frac 1 {(m+k)!} \frac {(\psi(1-r))^k q(t_{or})}
 {1- p_0(t_{or})}    \prod_{i=1}^{m-1} 2
\lambda q(x_i) \prod_{i=1}^m \frac{\psi(r + (1-r)p_0(y_i))}{q(y_i)}, \end{equation}
where the function $p_0(x)$ is the probability that an individual has no sampled
descendants for a time span of length $x$ so that $$p_0(x) = \frac { \lambda +
\mu + \psi + c_1 \frac {e^{-c_1x}(1-c_2) - (1+c_2)}{e^{-c_1x}(1-c_2)  +
(1+c_2)}}{2\lambda}$$ where $$c_1 = |\sqrt{(\lambda - \mu -\psi)^2 + 4 \lambda
\psi}|, \quad c_2 = - \frac{\lambda - \mu -\psi}{c_1}$$ and $$ q(x) = \frac {4} {2(1-c_2^2)
+ e^{-c_1x} (1-c_2)^2 + e^{c_1x}(1+c_2)^2}.$$ 
Throughout this paper, we consider non-oriented labeled trees. So equation~\eqref{epid} 
differs from the equation on page 350 in \cite{Stadler2011R0}, written for oriented trees, 
by a factor accounting for the switch from oriented to labeled trees  and also by the term 
for conditioning on $S$. Note also that the definition of the function $q$ here is different from
the definition in~\cite{Stadler2011R0}.

We show in the Supporting Information that function~\eqref{epid} depends only on
three parameters: $\lambda - \mu -\psi$, $\lambda \psi$, and $\psi (1-r)$, and
does not depend on parameters $\lambda$, $\mu$, $\psi$ and $r$ independently.
This means that the tree model is unidentifiable but, as we show in simulation
studies, if we specify one of the parameters we can estimate the others. 

When applying this model to  data, we typically shift time such that the most recent tip occurs at present, $y_m=0$, as 
we often do not have information about the length of time between the last sample 
and the end of the sampling effort. 
This is done to reduce our set of unknown quantities by one (namely setting $y_m=0$). 

We extend the model to allow the possibility of sampling individuals at present,
where each lineage at time 0 is sampled with probability $\rho$.  This process,  with $r$
set to zero, can be used to model speciation processes with fossilisation
events, hence it is called \emph{the fossilized birth-death
process}~\cite{Heath2013}. Let $S_\rho$ denote the event of sampling at least
one individual at present then according to~\cite{Stadler2010JTB} and accounting for labeled trees:
\begin{equation}\label{withrho0} f[g | \lambda, \mu, \psi, \rho, t_{or}, S_\rho]
= \frac 1 {(m+k)!}  \frac {\psi^k  \rho^n q(t_{or}) } {1 - \hat p_0(t_{or})} 
\prod_{i=1}^{m + n-1} 2 \lambda q(x_i) \prod_{i=1}^{m} \frac{\psi p_0(y_i)}{q(y_i)} \end{equation} where $n$ is the number of $\rho$-sampled
tips and 
$$\hat p_0(t_{or}) =1- \frac {\rho(\lambda - \mu)}{\lambda \rho +
(\lambda(1-\rho) - \mu) e^{-(\lambda-\mu)t}}.$$
In contrast to  the transmission birth-death process, where only three out of the four parameters $\lambda$, $\mu$, $\psi$, and $r$ can be inferred, under the fossilized birth-death process, all four parameters $\lambda$, $\mu$, $\psi$, and $\rho$ can be identified from the phylogeny.

It is possible to re-write density \eqref{withrho0} conditioning on the time of
the most recent common ancestor of sampled individuals rather than  conditioning
on the time of origin. 
In this case, we discard trees in which the root is a sampled node. 
In other words, we assume that the process starts with a
bifurcation event and we only consider trees with sampled nodes on both sides of
the initial bifurcation event. Then the time of the most recent common ancestor
of the sample is the time of the root, $x_1$. Accounting for labeled trees, the probability density function 
can thus be written \cite{Stadler2010JTB} as:

\begin{equation}\label{withrho1} f[g | \lambda, \mu, \psi, \rho, x_1, S_\rho]
= \frac 1 {(m+k)!}  \frac {\psi^k \rho^n q(x_1)} {\lambda(1 - \hat p_0(x_1))^2} \prod_{i=1}^{m + n-1} 2 \lambda q(x_i) \prod_{i=1}^{m} \frac{ \psi
p_0(y_i)}{q(y_i)}. \end{equation}

\subsubsection*{The sampled ancestor skyline model}

Here we extend the sampled ancestor birth-death  model so that  parameters may
change through time in a piecewise manner. This model combines two models from
\cite{Stadler2011R0} and \cite{StadKuhn12}.

Let  there be $\intNum$ time intervals $[t_i, t_{i-1})$ for $i \in \{1, \ldots, \intNum\}$ defined 
by vector $\bar t = (t_0, \ldots, t_{\intNum-1})$ and $t_\intNum=0$ with $t_\intNum < t_{\intNum-1} < \ldots < t_1 < t_0$ 
(where $t_0$ plays the role of the origin time, i.e., $t_{or}=t_0$).
We use notation $t_\intNum$ for time zero only for convenience and do not include it as   
a model parameter. 
Within each interval $[t_i, t_{i-1})$,  $i \in \{1, \ldots, \intNum\}$ the constant birth-death parameters
$\lambda_i$, $\mu_i$, $\psi_i$, and $r_i$ apply.  At the end of each interval at
times $t_i$, $i \in \{1, \ldots, \intNum\}$, each individual may be sampled with
probability $\rho_i$ (see also Figure~\ref{fig:fullsampledtree}). 
Thus, the model has $6\intNum$ parameters: $\bar \lambda$, $\bar
\mu$, $\bar \psi$, $\bar r$, $\bar \rho$, and $\bar t$. We prove in the
Supporting Information that the probability density of a reconstructed sampled
ancestor tree $g = (\mathcal T | \bar x, \bar y, \bar z)$ produced by this
process is (not conditioned on survival), 
\begin{equation} \label{4m}
\begin{multlined}
f[g |  \bar \lambda, \bar \mu, \bar \psi, \bar r, \bar \rho, \bar t] = \frac 1 {(m+M+k+K)!} \times \\
 q_1(t_0)  \prod_{i=1}^k (1-r_{\mathbf i_{z_i}})\psi_{\mathbf i_{z_i}}
\prod_{i=1}^{m + M-1}2 \lambda_{\mathbf i_{x_i}} q_{\mathbf i_{x_i}}(x_i)
\prod_{i=1}^{m}  \frac{\psi_{\mathbf i_{y_i}}(r_{\mathbf i_{y_i}} +
(1-r_{\mathbf i_{y_i}})p_{\mathbf i_{y_i}}(y_i))}{q_{\mathbf i_{y_i}}(y_i)}
\times \\
\prod_{i=1}^\intNum
((1-\rho_i)q_{i+1}(t_i))^{n_i}\rho_i^{N_i}((1-r_{i+1})q_{i+1}(t_i))^{K_i} (r_{i+1} +
(1-r_{i+1})p_{i+1}(t_i))^{M_i}
\end{multlined}
\end{equation} where $m$ is the number of $\psi$-sampled tips;
$k$ is the number of $\psi$-sampled nodes that have sampled descendants; $M_i$
is the number of tips sampled at time $t_i$;  $K_i$ is the number of nodes
sampled at time $t_i$ and having sampled descendants; $N_i = K_i + M_i$ is the
total number of nodes sampled at time $t_i$; $n_i$ is the number of lineages
present in the tree at time $t_i$ but not sampled at this time for $i~\in~\{1,
\ldots, \intNum\}$; $M = \sum\limits_{i=1}^\intNum M_i$; $K = \sum\limits_{i=1}^\intNum K_i$; ${\bf i}_x$ is an index such that
$t_{{\bf i}_x} \le x < t_{{\bf i}_x - 1}$; and functions $p_i$ and $q_i$ are
defined presently.

The probability $p_i(t)$ that an individual alive at time $t$ has no sampled
descendants when the process is stopped (i.e., in the time interval $[t_\intNum, t]$), with $t_i \le t <
t_{i-1}$ ($i=1, \ldots, \intNum$) is $$p_i(t) = \frac { \lambda_i + \mu_i + \psi_i -
A_i \frac {e^{A_i(t - t_i)} (1+B_i) - (1-B_i)}{e^{A_i(t - t_i)} (1+B_i) +
(1-B_i)} } {2 \lambda_i}$$ where $$A_i
= \sqrt {(\lambda_i - \mu_i - \psi_i)^2 + 4 \lambda_i \psi_i}$$ and $$ B_i =
\frac{(1-2(1-\rho_i)p_{i+1}(t_i))\lambda_i + \mu_i + \psi_i} {A_i}$$ for $i = 1,
\ldots, \intNum$ and $p_{l+1}(t_l) = 1$. Further, $$q_i(t) = \frac {4 e^{A_i(t-t_i)}}
{(e^{A_i(t-t_i)}(1+B_i) + (1-B_i))^2}$$ for $i = 1, \ldots, \intNum$. Note that
$q_{\intNum+1}(t_\intNum)$ does not appear in the equation because $n_\intNum$ (which is the
number of lineages present in the tree at time $t_\intNum$ but not sampled at that
time) and $K_\intNum$ (which is the number of two degree nodes at time $t_\intNum$) are always zero. 
Also, $r_{\intNum+1}$ cancels out because $K_\intNum$ is always zero and $p_{\intNum+1}(t_\intNum) = 1$.

We obtain two special cases of this general model that correspond to the skyline
variants of the transmission and fossilized birth-death processes by setting
some of the parameters to  zero.

To obtain the skyline transmission process, we set $\bar \rho = 0$. This 
implies $K_i=0$, $M_i=0$, and $N_i=0$ for all $i$. As before, we condition on
the event, $S$, of sampling at least one individual, where $f[S | \bar \lambda,
\bar \mu, \bar \psi, \bar t] = 1 - p_1(t_0)$. The tree density is
\begin{equation}\label{epidSky} 
\begin{multlined}
f[g |\bar \lambda, \bar \mu, \bar \psi, \bar r,
\bar t, S] = \frac 1 {(m+M+k+K)!} \times \\ \frac {q_1(t_0)}{1-p_1(t_0)}  \prod_{i=1}^k (1-r_{\mathbf
i_{z_i}})\psi_{\mathbf i_{z_i}} \prod_{i=1}^{m -1}2\lambda_{\mathbf i_{x_i}}
q_{\mathbf i_{x_i}}(x_i) \prod_{i=1}^{m}  \frac{\psi_{\mathbf
i_{y_i}}(r_{\mathbf i_{y_i}} + (1-r_{\mathbf i_{y_i}})p_{\mathbf
i_{y_i}}(y_i))}{q_{\mathbf i_{y_i}}(y_i)} \prod_{i=1}^\intNum (q_{i+1}(t_i))^{n_i} 
\end{multlined}
\end{equation}
 We show in the Supporting Information that \eqref{epidSky} can be 
re-parameterised with \begin{equation} \label{repar1} \begin{aligned} &d_i =
\lambda_i - \mu_i - \psi_i  & \text{ for  }&i = 1, \ldots, \intNum  \\ &f_i =
\lambda_i \psi_i  & \text{ for }&i = 1, \ldots, \intNum \\ & g_i = (1-r_i) \psi_i &
\text{ for }&i  = 1, \ldots, \intNum, \mbox{ and}\\ & k_i = \frac
{\lambda_i}{\lambda_{i+1}} & \text{ for }&i = 1, \ldots, \intNum - 1. \end{aligned}
\end{equation} 
Thus,  of the original $4\intNum$  parameters, only  $4\intNum-1$ may be
estimated.

For the skyline fossilized birth-death model, we set $\rho_1, \ldots,
\rho_{\intNum-1} = 0$ and $\bar r=0$ and condition on $S_\rho$,  the event of sampling
at least one extant individual (i.e., at time $t_\intNum$). The tree density becomes

\begin{equation}
\begin{multlined}\label{fossilSky}  f[g |\bar \lambda, \bar \mu, \bar \psi,
\rho_\intNum, \bar t, S_\rho]  = \frac 1 {(m+M+k+K)!} \times \\ \rho_\intNum^{N_\intNum} \frac{q_1(t_0)}{1- \hat
p_1(t_0)}  \prod_{i=1}^k \psi_{\mathbf i_{z_i}} \prod_{i=1}^{m + N_\intNum-1}2
\lambda_{\mathbf i_{z_i}} q_{\mathbf i_{x_i}}(x_i) \prod_{i=1}^{m} 
\frac{\psi_{\mathbf i_{y_i}}p_{\mathbf i_{y_i}}(y_i)}{q_{\mathbf i_{y_i}}(y_i)}
\prod_{i=1}^\intNum (q_{i+1}(t_i))^{n_i} 
\end{multlined} 
\end{equation} where $$ \hat p_1(t) = p_1(t |
\bar \psi = 0).$$ 
This probability density can be re-parameterised 
 as in~\eqref{repar1} with one additional equation $h = \lambda_\intNum \rho_\intNum$ (see  Supporting Information). Now there are $3\intNum+1$ initial parameters: $\bar \lambda$, $\bar \mu$, $\bar \psi$, and $\rho_\intNum$ and $4\intNum$ equations defining the re-parameterisation.
 Since $r_i =0$, $g_i$ defines $\psi_i$, then $f_i$ yields $\lambda_i$, then $d_i$ yields $\mu_i$,  $h$ yields $\rho_\intNum$ and the $\intNum-1$ equations for $k_i$ are not needed at all, thus $3\intNum+1$ equations define the reparameterization of the $3\intNum+1$ parameters, thus this re-parameterisation does not  reduce the number of parameters.  
 
\subsection*{Markov chain Monte Carlo Operators}

We introduce a number of operators  to explore  the space of sampled ancestor
trees with a fixed number of sampled nodes. Throughout this section, we denote
the height (or the age) of a node $a$ by $\tau_{a}$.


\subsubsection*{Extension of Wilson Balding operator}

We extend  the Wilson Balding  operator  (a type of subtree prune and
regraft)~\cite{WilsBald}  to sampled ancestor trees so that it is identical to 
the original operator when it is restricted to trees with no sampled ancestors.
The operator may propose a significant change to a tree and may change its
dimension, that is, the number of nodes in the tree.  We use the reversible jump
formalism of~\cite{Green1995}.

First, we describe a reduced version of the operator that does not change the
root. Let $g = (\mathcal T, \bar \tau)$ be a genealogy. There are three steps in
proposing a new tree. \begin{enumerate} \item Choose edge $e_1 = \lra{p_1, c_1}$
uniformly at random such that $p_1$ is not the root ($p_1$ is the parent of
$c_1$). Recall that we do not consider the origin as a node belonging to the tree.   
\item Choose either edge $e_2 = \lra{p_2, c_2}$ or leaf $l$. The method
of selection depends on the type of $e_1$: \begin{enumerate}  \label{pruning}
\item\label{branch} if node $c_1$ has a sibling then, uniformly at random from
all possibilities, either choose edge $e_2$ which is not adjacent to $e_1$ and
at least one end of which is above $c_1$ (i.e., $p_2$ is older than $c_1$) or
leaf $l$ which is older than $c_1$; \item\label{node} if node $c_1$  does not
have a sibling (so $p_1$ is a sampled node) then choose edge $e_2$ such that at
least one of its ends is older than $c_1$ or a leaf which is older than $c_1$
uniformly at random. \end{enumerate} If there is no such edge nor leaf,  do
nothing and propose no new tree. \item \label{attaching} If an item was chosen
in step 2, then prune the subtree rooted at node $p_1$ and reattach it to edge
$e_2$ or leaf $l$. When attaching to an edge, we draw a new height for the
parent of node $c_1$ uniformly at random from the interval
$[\max(\tau_{c_1},\tau_{c_2}),\tau_{p_2}]$.
\end{enumerate} Figure~\ref{fig: WBmove} illustrates pruning from a branch
(case~\ref{branch}) and from a node (case~\ref{node}) and attaching to a branch
and to a leaf. Let the resulting new genealogy be $g^* = (\mathcal T^*, \bar
\tau^*)$.

Now we extend this move to add the possibility of changing the root. We modify
the described procedure in two ways. First, we allow to choose $e_1$ for which
$p_1$ is the root at the first step. Second, we can also choose the root edge at
the second step, i.e., the edge which  connects the root with the origin. Although we do not usually consider this edge as a part of the tree, for convenience  we assume we can choose it. 
In this case, the parent of node $c_1$ becomes a new root with the
height obtained by drawing a difference between the new root height and the old
root height from the exponential distribution with rate $\lambda_e$.

To calculate the Hastings ratio, $\frac {q(g^* | g)}{q(g|g^*)}$, for this move we derive the proposal density, $q(g^* | g)$. $q(g^* | g)$    is a product of the probability of choosing
edge $e_1$ at the first step, the probability of choosing edge $e_2$ (or leaf
$l$) at the second step, and the probability density of choosing a new age at
the third stage (or one if we attach to a leaf).

Let $D$ denote the number of edges in tree $\mathcal T$. Then the contribution
of the first step to the proposal density is $\frac 1 D$. The probability at the second step 
 depends on the number of  choices there. However, since we choose the same subtree to prune in the forward and backward moves and then, at step two, choose from the items remaining in the tree after pruning the subtree,
the second terms in the product will cancel in the ratio and we do not calculate
them.

The contribution of the third step depends on the type of a move. When attaching
to a leaf it is equal to one. When attaching to a branch it is equal to the
probability density of a random variable $\tau^{new}$ which defines a new age
for the parent of $c_1$. So it is either $$f(\tau^{new}) = \frac 1 {|I_1|},
\text{ where $I_1 = (\tau_{p_2}, max \{\tau_{c_1}, \tau_{c_2} \})$}$$
or $$f(\tau^{new}) = \begin{cases} e^{-\lambda_eh_1},  &\text {if $h_1 =
\tau^{new} - \tau_1  > 0$}; \\ 0, &\text{otherwise.} \end{cases}$$ where
$\tau_{a}$ denotes the height of node $a$.  The  Hastings ratio for the different cases is summarised in Table~\ref{tab:hastings}.

\subsubsection*{Leaf to sampled ancestor jump}

This is a dimension changing move that jumps between two trees where a particular sampled node is a sampled ancestor in one tree and  a leaf in the other.
It randomly chooses a sampled node $i$.  
If $i$ is a sampled ancestor, we propose a new tree where $i$ is a leaf as follows.  Let $p$ be the parent of $i$ and $c$ be the child of $i$. Create a new node $j$ with  height chosen uniformly at random  from the interval $[\tau_i,\tau_p]$.  Make $p$ the parent of $j$ and make $i$ (now a leaf) and $c$ the children of $j$.  

If $i$ is a leaf then it becomes  a sampled ancestor replacing its parent if possible. It is not be possible if $i$ has no sibling or the sibling of $i$ is older than $i$. When this is possible, let node $b$ be the parent of $i$ in the proposed tree.   The Hastings ratio for this move is  $\frac 1 {\tau_p - \tau_i}$ when $i$ is a sampled ancestor and $(\tau_{b} - \tau_i)$ when $i$ is a leaf.

Note that these same trees can be proposed under  the extended Wilson Balding
operator. We introduce this more specific, or local, operator to improve  mixing. 

\subsubsection*{Other operators}

We  extend the narrow and wide exchange operators  used in BEAST2 to
sampled ancestor trees. The narrow exchange operator swaps a randomly chosen
node with its aunt if possible. It chooses a non-root
node $c$ such that its parent $p$ is not the root either. If the parent $b$ of
node $p$  is not a sampled node and, therefore, has another child $u$ and the
height of $u$ is less than the height of $c$  then we remove edges $\lra{p, c}$
and  $\lra{b, u}$ and add edges $\lra{p,u}$ and $\lra{b, c}$. Otherwise no tree
is proposed.  The wide exchange operator swaps two randomly chosen nodes along
with the subtrees descendant from these nodes if none of them is a parent to
another  one and the ages of the parents allow to swap the children. The Hastings
ratio is 1 for both operators. 

To propose height changes we use a scale operator and a uniform operator.  The scale operator scales non-sampled
internal nodes by a scale factor drawn from the uniform distribution on the
interval $(\frac 1 \beta, \beta)$, where $\beta > 1$. If the scaling makes some
parent node younger than its children then no tree is proposed. The Hastings ratio
for this operator is $\alpha^{k-2}$, where $\alpha$ is the scale factor
and $k$ is the number of internal non-sampled nodes (the number of scaled
dimensions).   The uniform  operator  proposes a new height for internal nodes chosen uniformly at random from the interval bounded by the heights of the parent and the oldest child of
the chosen node. The Hastings ratio for this operator is 1.

\subsection*{Simulations and empirical data analysis}

\subsubsection*{Simulating the fossilized birth-death process}

We simulated 100 trees under the sampled ancestor birth-death model with
$\rho$-sampling and $r=0$. We fix the tree model parameters in this simulation:

\begin{center} \begin{tabular}{p{2cm}p{2cm}}

$\lambda = 1.5$ & $t_{or} = 3.5$ \\ $\mu = 0.5$ & $\rho = 0.7$ \\ $\psi = 0.4$ 
&   \\ \end{tabular} \end{center} Since the time of the origin is one of the
model parameters, we simulate each tree on the time interval of $3.5$. We
discard trees with less than five sampled nodes, which constitute 8\% of the
trees. The remaining trees have 55 sampled nodes on average. 
Then we
simulated sequences along each tree under the GTR model with a strict molecular clock
model and ran the MCMC with the sequences and sampled node dates as the input data.
For these runs, we use the re-parameterisation: \renewcommand{\arraystretch}{1.2}
\begin{equation}\label{drs} \begin{tabular}{ll} \text{\emph{net diversification
rate}}  & $d = \lambda-\mu = 1.0$  \\ \text{\emph{turnover rate}} & $\turnover = 
\frac \mu \lambda = 0.33 $   \\ \text{\emph{\sampleprop}} & $\fosp = \frac
\psi {\mu + \psi}$ = 0.44  \\ \end{tabular} \end{equation}
along with the time of origin, $t_{or}$ and $\rho $.  The \sampleprop~is the proportion of individuals which are sampled before they are removed, 
meaning it is the proportion of sampled individuals out of all individuals in the full tree.
 Since this  set of parameters has only two parameters ($d$ and $t_{or}$) which are on the unbounded interval $(0,
\infty)$  with the others defined on  $[0,1]$, this is a convenient parametrisation for defining 
  uninformative priors. For the tree prior
distribution we use the distribution with probability density
function~\eqref{withrho0} multiplied by priors for hyper parameters: $\turnover$, $\fosp$,
and $\rho \sim$ Uniform(0,1) for  and  Uniform(0,1000) for $d$ and $t_{or}$.

 \renewcommand{\arraystretch}{1} 

We estimate a tree, tree model parameters, GTR rates, and the clock rate. The parameters of interest include tree model parameters ($d$, $\turnover$, $\fosp$
and $\rho$) and features of the tree including the time
of the origin ($t_{or}$), tree height and number of sampled ancestors.

\subsubsection*{Simulating the transmission birth-death process}

In this process, there is no $\rho$-sampling but $r>0$. Here we again use $d$,
$\turnover$, and $\fosp$ parametrisation defined by Equations~\eqref{drs}. We fix the time of the origin, $t_{or} = 3$, and
draw the tree model parameters from the distributions \begin{center}
\begin{tabular}{lcl} $d$ & $\sim$ & Uniform(1,2)   \\ $\turnover$ & $\sim$ &
Uniform(0,1)     \\ $\fosp$ & $\sim$ & Uniform(0.5,1)   \\ $r$ & $\sim$ &
Uniform(0,1) \\ \end{tabular} \end{center} 
and simulate a tree under the
transmission birth-death process with drawn parameters on the fixed time
interval. We choose these prior distributions because they cover a wide range of 
parameter combinations of interest and produce trees of reasonable size.
We discard trees with less than 5 or greater than 250 sampled nodes,
which constitute 21\% of the sample. In total, we report the results on 100
trees with the mean number of sampled nodes being 53. We simulate sequences along
each tree under the GTR model with a strict molecular clock.

In the MCMC runs, we fix the fossilisation proportion, $\fosp$, to its true value, as only three out of the four birth-death parameters can be inferred. The tree
prior distribution is~\eqref{epid} with uniform prior distributions for hyper
parameters $d$, $\turnover$, $\fosp$, and $r$, on the same intervals as above and
Uniform(0,1000) prior distribution for the time of the origin. We estimate the tree, tree model parameters, GTR rates and clock rate and assess the
estimates of the tree model parameters and properties of the tree.

\subsubsection*{Simulating under the sampled ancestor skyline model}

We  simulated  the skyline transmission process under three different sets of parameters and
then estimated the parameters in MCMC with fixed trees and with some parameters fixed. 
We have tried scenarios with two and three intervals, fixing either $r$ or $\psi$. 
In one scenario, only $\psi$ changes through time from zero to none-zero value and the other parameters 
stay constant. In the second  scenario, all parameters except $r$ change through time. 
In the final scenario, all parameters change through time and the whole vector $\bar r$ is fixed in
the inference. For a full description of the parameter and prior settings see the Supporting Information.

\subsubsection*{Bear dataset analysis}
We re-analyzed the bear dataset from \cite{Heath2013}. 
The  fossilized birth-death model we use is the same model as in the original analysis by Heath et
al.~\cite{Heath2013} but we use a strict clock  instead of a relaxed
clock model. We perform two analyses, both with a strict clock, using our implementation in 
BEAST2 and the implementation in  DPPDiv by
Heath et al. 

The tree prior density is~\eqref{withrho1} with transformed parameters $d$,
$\turnover$, and $\fosp$ for which we chose uniform priors and  $\rho=1$ is fixed.
We use the strict molecular clock model with and exponential prior for
the clock rate and the GTR model with gamma categories with uniform priors for
GTR rates and gamma shape.

The prior distributions in both analysis (in BEAST2 and DPPDiv) are all the same
except the priors for GTR rates and gamma shape. In  DPPDiv,
$$(\eta_{AC},\eta_{AG},\eta_{AT},\eta_{CG},\eta_{CT},\
\eta_{GT}) \sim \mbox{Dirichlet}(1,1,1,1,1,1)$$  In BEAST2, we fix
$\eta_{AG}$ to one and use Uniform(0, 100) priors for other rates.
We place a uniform prior for gamma shape parameter in BEAST2 and exponential in
DPPDiv.

\subsubsection*{HIV 1 dataset analysis}
We re-analyzed  UK HIV-1 subtype B data from \cite{Hue2005}.
We use the skyline model without $\rho$-sampling and with one rate shift time
(in 1999) because no samples were taken  before this time. The tree prior
density is~\eqref{epidSky}.  We use the following parameterisation and prior distributions:
 \renewcommand{\arraystretch}{1.2} \begin{center}
\begin{tabular}{llcl} \emph{effective reproductive number} & $R_0 = \frac
\lambda {\mu + \psi r}$ & $\sim$ & LogNormal(0.5,1)\\ \emph{total 
removal rate} & $\delta = \mu + \psi r$  & $\sim$ & LogNormal(-1,1)  \\
\emph{\leafsampleprop} & $\leafp = \frac {\psi r} {\mu + \psi r}$ & $\sim$ &
Uniform(0,1) \\ \emph{removal probability} & $r$ & $\sim$ & Beta(5,2) \\
\emph{time of origin} & $t_{or}$ & $\sim$ & LogNormal(3.28,0.5) \\
\end{tabular} \end{center} \renewcommand{\arraystretch}{1} 
The \leafsampleprop~is the proportion of individuals who are 
removed by sampling out of all removed individuals, thus
it is the proportion of sampled tips out of all tips in the full tree.
The parameterisation and prior distributions are different from 
the distributions used in simulation studies. We chose the prior distributions 
for $R_0$, $\delta$, and $\leafp$ following~\cite{StadKuhn12} and 
the prior distribution for $r$ assuming that diagnosed patients are likely to change 
their behaviour. Recall that this model is unidentifiable and we need to 
have a good prior knowledge about at least one of the parameters. 

We suppose that only \leafsampleprop{} changes through time and it changes from zero to a
non-zero value. Other parameters stay constant through time. We use a
GTR model with gamma categories and a molecular clock model with the
substitution rate fixed to $2.48 \times 10^{-3}$ as was estimated
in~\cite{StadKuhn12}.

\section*{Results}

We developed a Bayesian MCMC framework for phylogenetic inference with models
that allow sampled ancestors. We implemented a sampled ancestor MCMC algorithm  as
an add-on to software package BEAST2~\cite{Beast2} thereby making several sampled ancestor
birth-death prior models available to users. We test the accuracy and 
limitations of these models in simulation studies and apply the sampler to infer
divergence times for a biological dataset comprised of extant species and fossil
samples and to an HIV dataset. In case of the fossil-bear dataset, we compare the
results obtained from our implementation to the result obtained from an
alternative implementation \cite{Heath2013}.

\subsection*{Simulation of sampled ancestor models}

We simulated the sampled ancestor birth-death process and sampled ancestor skyline 
process under different scenarios. In all cases, the simulations show that we can
recover the tree and model parameters from sequence data and sampling
times.

For some variants of the model, one of the tree model parameters has to be fixed
for the inference to its true value as was discussed in the Methods section.
Simulation studies show that fixing one of the parameters allows to recover the
remaining
parameters. In particular, we showed that function~\eqref{epid}
depends exactly on three parameters because fixing $\psi$ allows recovery of
$\lambda$, $\mu$ and $r$ while function~\eqref{withrho0} depends on all four
parameters: $\lambda$, $\mu$, $\psi$ and $\rho$. We also simulated scenarios where
we fixed different parameters, for example, $r$ or $\psi$. All scenarios give
accurate estimates of the remaining parameters.

We present here detailed results of two sets of simulations: one
for the fossilized birth-death process and another one for the transmission
birth-death process. Further simulation results can be found in the 
Supporting Information.

In these two scenarios, we first simulated trees and then sequences along the
trees. Then we ran the sampler to recover tree model parameters and genealogies
from simulated data comprised of sequences and sampling times.  We assess the
results by calculating summary statistics including: the median estimate of a parameter, the error and relative bias of the median estimate, and
the relative 95\% highest posterior density (HPD) interval width. We assess whether the true value belongs
to the 95\% HPD interval. To summarise the results from a collection of runs we
calculate the medians of the summary statistics (i.e, the median of the
estimated medians, the median of the relative errors and so forth) and count the number
of times when the true value belongs to the 95\% HPD interval.  To assess the
power of the method with regard to estimation of sampled ancestors we performed
the receiver operating characteristic analysis~\cite{Swets} which estimates false positive and false negative error rates under different decision rules.

For the fossilized birth-death process (the process with $\rho$-sampling and
zero removal probability), we simulated a set of trees under a fixed set of the
tree model parameters. Each parameter was estimated and, in the worst case,  the  median of the relative errors for all runs was  0.22. The median of the relative errors for tree properties, such as the time of
origin, tree height and number of sampled ancestors, was at
most 0.09. The true parameters and tree properties were within the estimated 95\%
HPD intervals at least 95\% of the time in all cases. The estimates of the number of sampled
ancestors and the tree height are shown in Figure~\ref{fig: TreeChar}. Figure~\ref{fig: TurnoverHPD} shows how the amount of
uncertainty in turnover rate estimates decreases with the size of the tree
(i.e., with the number of sampled nodes).

To simulate from the transmission birth-death process, i.e., the
sampled ancestor birth-death process without $\rho$-sampling and with non-zero
removal probability, we draw tree model parameters from uniform
distributions for each simulation. The tree model parameters were estimated with
the maximum median of relative errors of 0.28 and, for the tree properties, of 0.06. In
the worst case a parameter or a tree property was inside the 95\% HPD
interval  92\% of the time. The estimates of the parameters are shown in
Figure~\ref{fig: Parameters}.

We used the data simulated from the transmission process to perform the receiver
operating characteristic (ROC) analysis of the sampled ancestor predictor, which makes
a prediction relying on the posterior distribution of genealogies. A node is
predicted to be a sampled ancestor with a probability calculated as a fraction
of trees in the posterior sample in which the node is a sampled ancestor. The
total number of non-final sampled nodes (we exclude the most recent node in each case
as it can not be a sampled ancestor) in all simulated trees was 5225 and
1814 of these nodes were sampled ancestors.  The ROC curve constructed from this
data and predictions obtained from the MCMC runs is shown in Figure~\ref{fig:
Roccurve}.

\subsection*{Application of the fossilized birth-death model to a bear dataset}

In~\cite{Heath2013}, Heath et al.  analysed a bear dataset comprised of sequence
data of 10 extant species and occurrence dates of 24 fossil samples  to estimate
divergence times in a Bayesian MCMC framework. They assumed that the tree
topology on the extant species is known and each fossil sample is assigned to a
clade in the tree. Here, we replicate this analysis using the MCMC
implementation of the fossilized birth-death model in BEAST2.  

We run two analysis with BEAST2 and with the DPPDiv implementation by Heath et
al. under the same model. The tree topology relating all living bear species and
two outgroup species is fixed in the analyses and we estimate the divergence
times and tree model parameters. The estimates are the same in both analyses as
expected.  The estimated divergence times are shown in Figure~\ref{fig:
Divergence}.  

\subsection*{Application of sampled ancestor Skyline model to HIV dataset}

We analysed an HIV-1 subtype B dataset from the United Kingdom, consisting of 62
sequences that were originally analysed in~\cite{Hue2005} and later analysed
using the skyline model without sampled ancestors in~\cite{StadKuhn12}. The
posterior probability of being a sampled ancestors for three sampled nodes was 
 61\%, 59\%, and 49\%. For other sampled nodes the probability was less than 4\%. 
There is positive evidence that three sampled nodes with high posterior probabilities 
are sampled ancestors. The Bayes factors are 5.9, 8.7, and 4.2, respectively. 

We chose a random tree among the trees in the posterior sample that have exactly
these three nodes as sampled ancestors. The tree is shown in Figure~\ref{fig:
RandomTree}. It is noticeable that all three sampled ancestors are clustered on
a clade of 16 (out of 62) samples. The median of the posterior distribution of
the number of sampled ancestors was 2 with 95\% HPD interval $[1, 3]$. The
removal probability was estimated to be 0.74 with 95\% HPD interval $[0.46,
0.97]$.

\section*{Discussion}

The MCMC sampler developed here enables analyses under models in which the
probability of one sample being the direct ancestor of another sample is not negligible.
These models are useful for describing infectious transmission processes,
including identifying transmission chains. They are
also useful for estimating divergence times for macroevolutionary data in the
presence of fossil samples.

In the analysis of a phylogeny of bears we show that the sampler can be applied
to data comprised of both fossil and recent taxa to infer divergence times. This
dataset was previously analysed using the {\it fossilized birth-death model} by
Heath et al. \cite{Heath2013}. While the underlying model is the same and thus producing the same results,
there is a conceptual difference between the
two MCMC frameworks. In the analysis by Heath {\it et al}, MCMC was used to
integrate over fossil attachment times while the topological attachment of the
fossils was integrated out analytically. To achieve this, the topology
of the phylogeny relating the extant taxa had to be assumed to be known. In our
implementation, we integrate  over the trees relating fossil and extant taxa,
i.e., over both the fossil attachment times and topological attachment points, using
MCMC. In order to facilitate a direct comparison we  constrained the
topology of the extant species, however the sampler does not require this.
For datasets where the tree topology is well resolved, analytical
calculation results in faster mixing but when there is uncertainty in the extant
phylogeny, which is the more common case, our sampler can account for it.
Since the two implementations of the method were made completely independently of one another, 
this result also provides strong evidence that our implementation is sampling from the correct posterior distribution.

A natural extension to the analysis of the bear phylogeny would be to include
morphological data to inform the inference regarding the precise placement of
fossils on the tree \cite{Pyron2011, Wood2012},  however this requires models of
morphological character evolution \cite{Lewis2001,Ronq2012}. Another direction
for application of the sampler is using the skyline version of the fossilized
birth-death model to analyse datasets where fossil samples come from different
stratigraphic layers, so that rates of fossilisation and discovery may change
through time. Fossils are better preserved in some layers than in other layers and
therefore the sampling rate varies from layer to layer and this can be modelled
as a skyline plot.

Simulation studies show that the MCMC sampler for sampled ancestor trees allows
for the detection of direct ancestors within the sample. In epidemiological
studies, sampled ancestors can be interpreted as sampled individuals that have later infected
other individuals. In the analysis of the HIV dataset, we equated the
transmission tree directly with the viral gene tree. This approximation is good
enough to demonstrate the method. But for chronic infectious diseases such as
Hepatitis C and HIV where the genetic diversity of the pathogen population
within a single host can be substantial (e.g.
\cite{Shankarappa1999,Vrancken2014}) the inferential power would be improved by
a hierarchical model that explicitly models the difference between the sampled
ancestor) transmission tree and the (binary) viral gene tree. Regardless of the
modelling details, such analyses allow for the estimation of the removal at sampling parameter $r$,
which controls the prevalence of sampled ancestors. In most situations this
parameter reflects the probability with which patients remain able to cause
further infections after they were diagnosed.

Analytic calculations (presented in the Supporting Information) and
simulation studies show that there is a degree of non-identifiably of parameters
in the transmission birth-death models that include the $r$ parameter. In other
words, these models require  one of the parameters to be fixed or strongly constrained by prior information to
achieve unambiguous inference. In epidemiological studies with a known sampling scheme, a candidate parameter to fix is the \sampleprop{}. 
For epidemics with a period of infection, such as influenza, 
the total removal rate, $\delta$, could be fixed. 
Under the fossilized birth-death model,  it is possible to infer all the parameters of the tree
process prior when time-stamped comparative data is available. 
This is an interesting insight: if no fossils are available, we can only infer two out of the three parameters $\lambda,\mu,\rho$ (as the likelihood only depends on $\lambda-\mu, \lambda \rho$) while in presence of fossils we can estimate all four parameters  $\lambda,\mu,\rho,\psi$ (as the likelihood  depends on $\lambda-\mu, \lambda \rho, \lambda \psi, \psi$).
As sequence
data of fossil organisms is rarely available and thus information about fossil
locations on the tree obtained by phylogenetic modelling of morphological data
\cite{Lewis2001,Ronq2012} may become important to enable effective inference.

To our knowledge this is the first full implementation of an MCMC sampler of 
sampled ancestor trees and we anticipate that such samplers will form the 
computational basis for further developments in both fossil-calibrated divergence time dating
and phylodynamics.

\section*{Acknowledgments}

We would like to thank Tracy Heath for assistance with her implementation. We
also acknowledge the New Zealand eScience Infrastructure (NeSI) for use of their
high-performance computing facilities. AJD was funded by a Rutherford Discovery
Fellowship from the Royal Society of New Zealand. This research was also
partially supported by Marsden grant \#UOA1324 from the Royal Society of New
Zealand
(http://www.royalsociety.org.nz/programmes/funds/marsden/awards/2013-awards/).

\bibliography{sampledAncestorMCMC}

\begin{thebibliography}{10}
\providecommand{\url}[1]{\texttt{#1}}
\providecommand{\urlprefix}{URL }
\expandafter\ifx\csname urlstyle\endcsname\relax
  \providecommand{\doi}[1]{doi:\discretionary{}{}{}#1}\else
  \providecommand{\doi}{doi:\discretionary{}{}{}\begingroup
  \urlstyle{rm}\Url}\fi
\providecommand{\bibAnnoteFile}[1]{%
  \IfFileExists{#1}{\begin{quotation}\noindent\textsc{Key:} #1\\
  \textsc{Annotation:}\ \input{#1}\end{quotation}}{}}
\providecommand{\bibAnnote}[2]{%
  \begin{quotation}\noindent\textsc{Key:} #1\\
  \textsc{Annotation:}\ #2\end{quotation}}
\providecommand{\eprint}[2][]{\url{#2}}

\bibitem{Yang1997}
Yang Z, Rannala B (1997) {B}ayesian phylogenetic inference using {DNA}
  sequences: a {Markov chain Monte Carlo} method.
\newblock Mol Biol Evol 14: 717-24.
\bibAnnoteFile{Yang1997}

\bibitem{Mau1999}
Mau B, Newton MA, Larget B (1999) {B}ayesian phylogenetic inference via {Markov
  chain Monte Carlo} methods.
\newblock Biometrics 55: 1-12.
\bibAnnoteFile{Mau1999}

\bibitem{MrBayes}
Huelsenbeck JP, Ronquist F (2001) {MRBAYES}: {B}ayesian inference of
  phylogenetic trees.
\newblock Bioinformatics 17: 754-5.
\bibAnnoteFile{MrBayes}

\bibitem{Beast17}
Drummond AJ, Suchard MA, Xie D, Rambaut A (2012) {B}ayesian phylogenetics with
  {BEAUti} and the {BEAST} 1.7.
\newblock Mol Biol Evol 29: 1969-73.
\bibAnnoteFile{Beast17}

\bibitem{MrBayes32}
Ronquist F, Teslenko M, van~der Mark P, Ayres DL, Darling A, et~al. (2012)
  {MrBayes} 3.2: efficient {B}ayesian phylogenetic inference and model choice
  across a large model space.
\newblock Syst Biol 61: 539--42.
\bibAnnoteFile{MrBayes32}

\bibitem{Beast2}
Bouckaert R, Heled J, K\"{u}nert D, Vaughan TG, Wu CH, et~al. (2014) {BEAST2}:
  A software platform for {B}ayesian evolutionary analysis.
\newblock PLoS Comput Biol 10: e1003537.
\bibAnnoteFile{Beast2}

\bibitem{PhyloBayes3}
Lartillot N, Lepage T, Blanquart S (2009) {PhyloBayes} 3: a {B}ayesian software
  package for phylogenetic reconstruction and molecular dating.
\newblock Bioinformatics 25: 2286-8.
\bibAnnoteFile{PhyloBayes3}

\bibitem{Metropolis1953}
Metropolis N, Rosenbluth A, Rosenbluth M, Teller A, Teller E (1953) Equations
  of state calculations by fast computing machines.
\newblock Journal of Chemistry and Physics 21: 1087--1092.
\bibAnnoteFile{Metropolis1953}

\bibitem{Hastings1970}
Hastings W (1970) {Monte Carlo} sampling methods using {Markov} chains and
  their applications.
\newblock Biometrika 57: 97--109.
\bibAnnoteFile{Hastings1970}

\bibitem{Lewis2005}
Lewis PO, Holder MT, Holsinger KE (2005) Polytomies and {B}ayesian phylogenetic
  inference.
\newblock Syst Biol 54: 241--253.
\bibAnnoteFile{Lewis2005}

\bibitem{Gavr2013}
Gavryushkina A, Welch D, Drummond AJ (2013) Recursive algorithms for
  phylogenetic tree counting.
\newblock Algorithms for Molecular Biology 8: 26.
\bibAnnoteFile{Gavr2013}

\bibitem{Drummond2003}
Drummond A, Pybus O, Rambaut A, Forsberg R, Rodrigo A (2003) Measurably
  evolving populations.
\newblock Trends in Ecology \& Evolution 18: 481--488.
\bibAnnoteFile{Drummond2003}

\bibitem{Pybus2001}
Pybus OG, Charleston MA, Gupta S, Rambaut A, Holmes EC, et~al. (2001) The
  epidemic behavior of the hepatitis {C} virus.
\newblock Science 292: 2323-2325.
\bibAnnoteFile{Pybus2001}

\bibitem{Gren2004}
Grenfell BT, Pybus OG, Gog JR, Wood JLN, Daly JM, et~al. (2004) Unifying the
  epidemiological and evolutionary dynamics of pathogens.
\newblock Science 303: 327-32.
\bibAnnoteFile{Gren2004}

\bibitem{Stadler2011R0}
Stadler T, Kouyos RD, von Wyl V, Yerly S, B\"{o}ni J, et~al. (2011) Estimating
  the basic reproductive number from viral sequence data.
\newblock Molecular Biology and Evolution 29: 347-357.
\bibAnnoteFile{Stadler2011R0}

\bibitem{Ypma13}
Ypma RJ, van Ballegooijen WM, Wallinga J (2013) Relating phylogenetic trees to
  transmission trees of infectious disease outbreaks.
\newblock Genetics 195: 1055-1062.
\bibAnnoteFile{Ypma13}

\bibitem{Kuhnert2014}
K\"{u}nert D, Stadler T, Vaughan TG, Drummond AJ (2014) Simultaneous
  reconstruction of evolutionary history and epidemiological dynamics from
  viral sequences with the birth-death {SIR} model.
\newblock J R Soc Interface 11: 20131106.
\bibAnnoteFile{Kuhnert2014}

\bibitem{volz2013inferring}
Volz EM, Frost SD (2013) Inferring the source of transmission with phylogenetic
  data.
\newblock PLoS computational biology 9: e1003397.
\bibAnnoteFile{volz2013inferring}

\bibitem{Teunis13}
Teunis P, Heijne JCM, Sukhrie F, van Eijkeren J, Koopmans M, et~al. (2013)
  Infectious disease transmission as a forensic problem: who infected whom?
\newblock J R Soc Interface 10: 20120955.
\bibAnnoteFile{Teunis13}

\bibitem{Vrancken2014}
Vrancken B, Rambaut A, Suchard MA, Drummond A, Baele G, et~al. (2014) The
  genealogical population dynamics of {HIV}-1 in a large transmission chain:
  Bridging within and among host evolutionary rates.
\newblock PLoS Comput Biol 10: e1003505.
\bibAnnoteFile{Vrancken2014}

\bibitem{Stadler2010JTB}
Stadler T (2010) Sampling-through-time in birth-death trees.
\newblock Journal of Theoretical Biology 267: 396-404.
\bibAnnoteFile{Stadler2010JTB}

\bibitem{Didier12}
Didiera G, Royer-Carenzib M, Laurinc M (2012) The reconstructed evolutionary
  process with the fossil record.
\newblock J Theor Biol 315: 26--37.
\bibAnnoteFile{Didier12}

\bibitem{StadKuhn12}
Stadler T, K\"{u}nert D, Bonhoeffer S, Drummond AJ (2013) Birth-death skyline
  plot reveals temporal changes of epidemic spread in {HIV} and hepatitis {C}
  virus ({HCV}).
\newblock Proc Natl Acad Sci USA 110: 228--33.
\bibAnnoteFile{StadKuhn12}

\bibitem{Sanderson1997}
Sanderson M (1997) A nonparametric approach to estimating divergence times in
  the absence of rate consistency.
\newblock Molecular Biology and Evolution 14: 1218--1231.
\bibAnnoteFile{Sanderson1997}

\bibitem{Thorne1998}
Thorne JL, Kishino H, Painter IS (1998 Dec) Estimating the rate of evolution of
  the rate of molecular evolution.
\newblock Mol Biol Evol 15: 1647--1657.
\bibAnnoteFile{Thorne1998}

\bibitem{Drummond2006}
Drummond AJ, Ho SYW, Phillips MJ, Rambaut A (2006) Relaxed phylogenetics and
  dating with confidence.
\newblock PLoS Biol 4: e88.
\bibAnnoteFile{Drummond2006}

\bibitem{Rannala2007}
Rannala B, Yang Z (2007) Inferring speciation times under an episodic molecular
  clock.
\newblock Systematic Biology 56: 453--466.
\bibAnnoteFile{Rannala2007}

\bibitem{Ho2009}
Ho SYW, Phillips MJ (2009) Accounting for calibration uncertainty in
  phylogenetic estimation of evolutionary divergence times.
\newblock Syst Biol 58: 367--380.
\bibAnnoteFile{Ho2009}

\bibitem{Ronq2012}
Ronquist F, Klopfstein S, Vilhelmsen L, Schulmeister S, Murray DL, et~al.
  (2012) A total-evidence approach to dating with fossils, applied to the early
  radiation of the hymenoptera.
\newblock Syst Biol 61: 973-999.
\bibAnnoteFile{Ronq2012}

\bibitem{Heath2013}
Heath T, Huelsenbeck J, Stadler T (2013) The fossilised birth-death process: A
  coherent model of fossil calibration for divergence time estimation.
\newblock Proc Natl Acad Sci U S A : In press.
\bibAnnoteFile{Heath2013}

\bibitem{Heled2012}
Heled J, Drummond AJ (2012) Calibrated tree priors for relaxed phylogenetics
  and divergence time estimation.
\newblock Syst Biol 61: 138-49.
\bibAnnoteFile{Heled2012}

\bibitem{Pyron2011}
Pyron RA (2011) Divergence time estimation using fossils as terminal taxa and
  the origins of {L}issamphibia.
\newblock Syst Biol 60: 466--81.
\bibAnnoteFile{Pyron2011}

\bibitem{Wood2012}
Wood HM, Matzke NJ, Gillespie RG, Griswold CE (2013) Treating fossils as
  terminal taxa in divergence time estimation reveals ancient vicariance
  patterns in the palpimanoid spiders.
\newblock Syst Biol 62: 264--284.
\bibAnnoteFile{Wood2012}

\bibitem{Silv2014}
Silvestro D, Schnitzler J, Liow LH, Antonelli A, Salamin N (2014) {B}ayesian
  estimation of speciation and extinction from incomplete fossil occurrence
  data.
\newblock Syst Biol 63: 349-367.
\bibAnnoteFile{Silv2014}

\bibitem{Foote1996}
Foote M (1996) On the probability of ancestors in the fossil record.
\newblock Paleobiology 22: 141--151.
\bibAnnoteFile{Foote1996}

\bibitem{Green1995}
Green PJ (1995) Reversible jump {Markov chain Monte Carlo} computation and
  {B}ayesian model determination.
\newblock Biometrica 82: 771--732.
\bibAnnoteFile{Green1995}

\bibitem{WilsBald}
Wilson IJ, Balding DJ (1998) Genealogical inference from microsatellite data.
\newblock Genetics 150: 499-510.
\bibAnnoteFile{WilsBald}

\bibitem{Hue2005}
Hu\'{e} S, Pillay D, Clewley JP, Pybus OG (2005) Genetic analysis reveals the
  complex structure of {HIV}-1 transmission within defined risk groups.
\newblock Proc Natl Acad Sci USA 102: 4425-4429.
\bibAnnoteFile{Hue2005}

\bibitem{Swets}
Swets JA (1996) Signal detection theory and ROC analysis in psychology and
  diagnostics : collected papers.
\newblock Lawrence Erlbaum Associates, Mahwah, NJ.
\bibAnnoteFile{Swets}

\bibitem{Lewis2001}
Lewis PO (2001) A likelihood approach to estimating phylogeny from discrete
  morphological character data.
\newblock Syst Biol 50: 913-25.
\bibAnnoteFile{Lewis2001}

\bibitem{Shankarappa1999}
Shankarappa R, Margolick J, Gange S, Rodrigo A, Upchurch D, et~al. (1999)
  Consistent viral evolutionary changes associated with the disease progression
  of human immunodeficiency virus type 1 infection.
\newblock Journal virology 73: 10489--10502.
\bibAnnoteFile{Shankarappa1999}

\end{thebibliography}

\newpage

\section*{Figures}
%

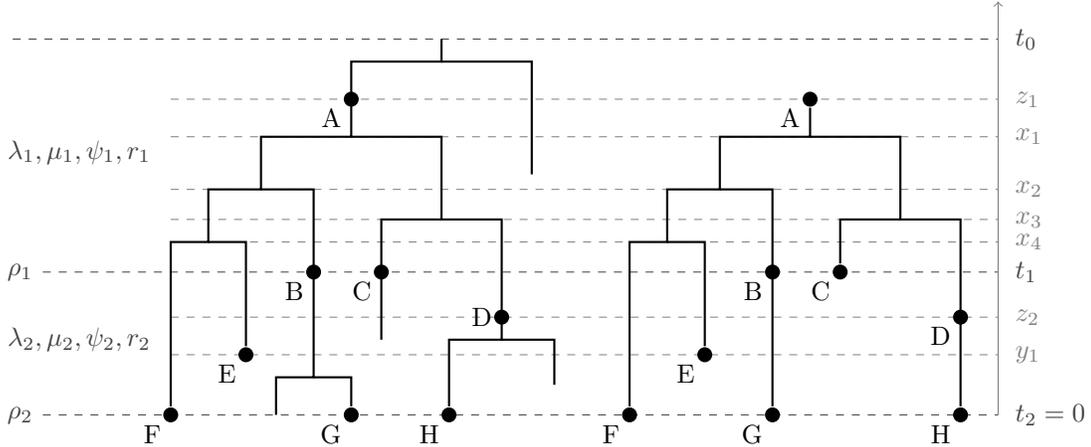
\begin{figure}[H] \begin{center} 
\usetikzlibrary{shapes,snakes}

\begin{tikzpicture}[thick]

\color{darkgray}


\draw[dashed, very thin] (-5.1, 7) -- (8, 7); 
\node[right] at (8.1, 7){$t_0$};
\draw[dashed, very thin] (-4.7, 3.9) -- (8, 3.9); 
\node[right] at (8.1, 3.9){$t_1$};
\draw[dashed, very thin] (-4.7, 2) -- (8, 2); 
\node[right] at (8.1, 2){$t_2 = 0$};

\node[right] at (-5.3, 5.5) {$\lambda_1, \mu_1, \psi_1, r_1$};
\node[right] at (-5.3, 3.0){$\lambda_2, \mu_2, \psi_2, r_2$};
\node[left] at (-4.7, 3.9){$\rho_1$};
\node[left] at (-4.7, 2){$\rho_2$};

\color{gray}

\draw[dashed, very thin] (-3, 6.2) -- (8, 6.2); 
\node[right] at (8.1, 6.2){$z_1$};
\draw[dashed, very thin] (-3, 5.7) -- (8, 5.7); 
\node[right] at (8.1, 5.7){$x_1$};
\draw[dashed, very thin] (-3, 5.0) -- (8, 5); 
\node[right] at (8.1, 5){$x_2$};
\draw[dashed, very thin] (-3, 4.6) -- (8, 4.6); 
\node[right] at (8.1, 4.6){$x_3$};
\draw[dashed, very thin] (-3, 4.3) -- (8, 4.3); 
\node[right] at (8.1, 4.3){$x_4$};
\draw[dashed, very thin] (-3, 3.3) -- (8, 3.3); 
\node[right] at (8.1, 3.3){$z_2$};
\draw[dashed, very thin] (-3, 2.8) -- (8, 2.8); 
\node[right] at (8.1, 2.8){$y_1$};

\draw[thin, arrows = ->, color=gray] (8, 2) -- (8, 7.5);

\color{black}

\begin{scope}[xshift=0.9cm]

\node[fill,circle, inner sep=2pt]  at (-1.5, 6.2)(sA){};
\node[anchor=north east] (Alet) at (sA) {\textcolor{black}{A}};
\node[fill,circle, inner sep=2pt]  at (-2, 3.9)(sB){};
\node[anchor=north east] (Alet) at (sB) {\textcolor{black}{B}};
\node[fill,circle, inner sep=2pt]  at (-1.1, 3.9)(sC){};
\node[anchor=north east] (Alet) at (sC) {\textcolor{black}{C}};
\node[fill,circle, inner sep=2pt]  at (0.5, 3.3)(sD){};
\node[anchor=east] (Alet) at (sD) {\textcolor{black}{D}};
\node[fill,circle, inner sep=2pt]  at (-2.9, 2.8)(sE){};
\node[anchor=north east] (Alet) at (sE) {\textcolor{black}{E}};

\node[fill,circle, inner sep=2pt]  at (-3.9, 2)(sF){};
\node[anchor=north east] (Alet) at (sF) {\textcolor{black}{F}};
\node[fill,circle, inner sep=2pt]  at (-1.5, 2)(sG){};
\node[anchor=north east] (Alet) at (sG) {\textcolor{black}{G}};
\node[fill,circle, inner sep=2pt]  at (-0.2, 2)(sH){};
\node[anchor=north east] (Alet) at (sH) {\textcolor{black}{H}};

\draw (-0.3, 7) -- (-0.3, 6.7) -- (-1.5, 6.7) -- (-1.5, 5.7) -- (-2.7, 5.7) -- (-2.7, 5.0) -- (-3.4, 5.0) -- (-3.4, 4.3) -- (-3.9, 4.3) -- (sF) ;
\draw (-3.4, 4.3) -- (-2.9, 4.3) -- (sE); 
\draw (-2.7, 5.0)  -- (-2, 5.0) -- (-2, 2.5) -- (-2.5, 2.5) -- (-2.5, 2); 
\draw (-2, 2.5) -- (-1.5, 2.5) -- (sG);
\draw (-1.5, 5.7) -- (-0.3, 5.7) -- (-0.3, 4.6) -- (-1.1, 4.6) -- (-1.1, 3);
\draw (-0.3, 4.6) -- (0.5, 4.6) -- (0.5, 3) -- (-0.2, 3) -- (sH);
\draw (0.5, 3) -- (1.2, 3) -- (1.2, 2.4);
\draw (-0.3, 6.7) -- (0.9, 6.7) -- (0.9, 5.2);

\end{scope}

\begin{scope}[xshift=7cm]

\node[fill,circle, inner sep=2pt]  at (-1.5, 6.2)(sA){};
\node[anchor=north east] (Alet) at (sA) {\textcolor{black}{A}};
\node[fill,circle, inner sep=2pt]  at (-2, 3.9)(sB){};
\node[anchor=north east] (Alet) at (sB) {\textcolor{black}{B}};
\node[fill,circle, inner sep=2pt]  at (-1.1, 3.9)(sC){};
\node[anchor=north east] (Alet) at (sC) {\textcolor{black}{C}};
\node[fill,circle, inner sep=2pt]  at (0.5, 3.3)(sD){};
\node[anchor=north east] (Alet) at (sD) {\textcolor{black}{D}};
\node[fill,circle, inner sep=2pt]  at (-2.9, 2.8)(sE){};
\node[anchor=north east] (Alet) at (sE) {\textcolor{black}{E}};

\node[fill,circle, inner sep=2pt]  at (-3.9, 2)(sF){};
\node[anchor=north east] (Alet) at (sF) {\textcolor{black}{F}};
\node[fill,circle, inner sep=2pt]  at (-2, 2)(sG){};
\node[anchor=north east] (Alet) at (sG) {\textcolor{black}{G}};
\node[fill,circle, inner sep=2pt]  at (0.5, 2)(sH){};
\node[anchor=north east] (Alet) at (sH) {\textcolor{black}{H}};

\draw (sA) -- (-1.5, 5.7) -- (-2.7, 5.7) -- (-2.7, 5.0) -- (-3.4, 5.0) -- (-3.4, 4.3) -- (-3.9, 4.3) -- (sF) ;
\draw (-3.4, 4.3) -- (-2.9, 4.3) -- (sE); 
\draw (-2.7, 5.0)  -- (-2, 5.0) -- (sG);
\draw (-1.5, 5.7) -- (-0.3, 5.7) -- (-0.3, 4.6) -- (-1.1, 4.6) -- (sC);
\draw (-0.3, 4.6) -- (0.5, 4.6) -- (sH);

\end{scope}

\end{tikzpicture} 
\caption{
{\bf Full tree versus reconstructed tree.} A full tree produced by the sampled ancestor birth-death process on the left and a reconstructed tree on the right. The sampled nodes are indicated by dots labeled by letters A through H.  Nodes A, B and D are sampled ancestors.
The reconstructed tree is represented by a sampled ancestor tree $g = (\mathcal T,
(x_1, x_2, x_3, x_4, y_1, z_1, z_2))$, where $\mathcal T$ denotes the ranked tree topology and
$\bar x$, $\bar y$, and $\bar z$ denote the node ages. In the reconstructed tree the root is a sampled node. 
In the skyline model, birth-death parameters vary from interval to interval. There are two intervals in this figure bounded by 
the time of origin $t_0$, parameter shift time $t_1$, and present time $t_2$. Between $t_0$ and $t_1$ parameters 
$\lambda_1$, $\mu_1$, $\psi_1$ and $r_1$ apply and between $t_1$ and $t_2$ parameters $\lambda_2$, $\mu_2$, $\psi_2$, and $r_2$.
There are additional sampling attempts at times $t_1$ and $t_2$ with sampling probabilities $\rho_1$ and $\rho_2$.} \label{fig:fullsampledtree} \end{center}
\end{figure}

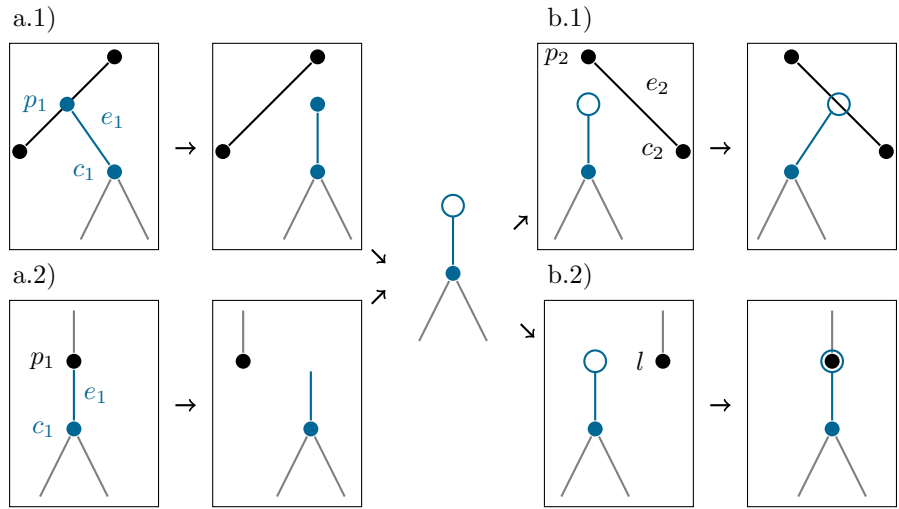
\begin{figure}[H] \begin{center}
\usetikzlibrary{decorations.markings}

\begin{tikzpicture}[thick, scale=0.9]

\begin{scope}

\node[right] at (0.05, 3.25) {a.1)};

\draw[thin] (0.15,-0.15) rectangle (2.35,2.9);
\node[fill,circle, inner sep=2pt, color=black] at (0.3,1.3) (CiP) {};
\node[fill,circle, inner sep=2pt, color=black] at (1.7,2.7) (PiP) {};
\draw[black] (CiP) -- (PiP);

\node[fill,circle, inner sep=2pt, color=MidnightBlue, label=left:\textcolor{MidnightBlue}{$c_1$}] at (1.7,1) (i) {};
\node[fill,circle, inner sep=2pt, color=MidnightBlue, label=left:\textcolor{MidnightBlue}{$p_1$}] at (1,2) (iP) {};
\draw[MidnightBlue] (i) -- node [above right] {$e_1$}  (iP);

\draw[gray] (1.2,0) -- (i) -- (2.2,0);

\draw[arrows = ->] (2.6,1.35) -- (2.9,1.35);

\draw[thin] (3.15,-0.15) rectangle (5.35,2.9);
\node[fill,circle, inner sep=2pt, color=black] at (3.3,1.3) (CiP) {};
\node[fill,circle, inner sep=2pt, color=black] at (4.7,2.7) (PiP) {};

\draw[black] (CiP) -- (PiP);

\node[fill,circle, inner sep=2pt, color=MidnightBlue] at (4.7,1) (i) {};
\node[fill,circle, inner sep=2pt, color=MidnightBlue] at (4.7,2) (iP) {};
\draw[MidnightBlue] (i) -- (iP);

\draw[gray] (4.2,0) -- (i) -- (5.2,0);

\draw[arrows = ->] (5.5,-0.15) -- (5.7,-0.35);

\end{scope}

\begin{scope}[yshift=-3.8cm]

\node[right] at (0.05, 3.25) {a.2)};

\draw[thin] (0.15,-0.15) rectangle (2.35,2.9);
\node[fill, circle, inner sep=2pt, color=black, label=left:{$p_1$}] at (1.1,2) (iP) {};
\draw[gray] (iP) -- (1.1, 2.75);

\node[fill,circle, inner sep=2pt, color=MidnightBlue,  label=left:\textcolor{MidnightBlue}{$c_1$}] at (1.1,1) (i) {};
\draw[MidnightBlue] (i) -- node [right] {$e_1$} (iP);

\draw[gray] (0.6,0) -- (i) -- (1.6,0);

\draw[arrows = ->] (2.6,1.35) -- (2.9,1.35);

\draw[thin] (3.15,-0.15) rectangle (5.35,2.9);
\node[fill,circle, inner sep=2pt, color=black] at (3.6,2) (iP) {};
\draw[gray] (iP) -- (3.6, 2.75);

\node[fill,circle, inner sep=2pt, color=MidnightBlue] at (4.6,1) (i) {};
\draw[MidnightBlue] (i) -- (4.6,1.85);

\draw[gray] (4.1,0) -- (i) -- (5.1,0);

\draw[arrows = ->] (5.5,2.85) -- (5.7,3.05);
\end{scope}


\begin{scope}[xshift=6.2cm]

\draw[MidnightBlue] (0.5,0.5) circle [radius=0.16];
\node[fill,circle, inner sep=2pt, color=MidnightBlue] at (0.5,-0.5) (i) {};
\draw[MidnightBlue] (i) -- (0.5,0.35);

\draw[gray] (0,-1.5) -- (i) -- (1,-1.5);

\end{scope}


\begin{scope}[xshift=11.2cm, yshift=-3.5cm]

\draw[arrows = ->] (-3.6,3.65) -- (-3.4,3.85);

\node[right] at (-3.25, 6.75) {b.1)};

\draw[thin] (-3.25, 3.35) rectangle (-0.95,6.4);
\node[fill,circle, inner sep=2pt, color=black, label=left:$c_2$] at (-1.1,4.8) (j) {};
\node[fill,circle, inner sep=2pt, color=black, label=left:$p_2$] at (-2.5,6.2) (jP) {};
\draw[black] (j)-- node[above right]{$e_2$} (jP);

\node[fill,circle, inner sep=2pt, color=MidnightBlue] at (-2.5,4.5) (i) {};
\draw[MidnightBlue] (-2.5,5.5) circle [radius=0.16];
\draw[MidnightBlue] (i)--(-2.5,5.35); 

\draw[gray] (-3,3.5) -- (i) -- (-2,3.5);

\draw[arrows = ->] (-0.7,4.85) -- (-0.4,4.85);

\draw[thin] (-0.15, 3.35) rectangle (2.05,6.4);
\node[fill,circle, inner sep=2pt, color= black] at (1.9,4.8) (j) {};
\node[fill,circle, inner sep=2pt, color=black] at (0.5,6.2) (jP) {};
\draw[black] (j)--(jP);

\node[fill,circle, inner sep=2pt, color=MidnightBlue] at (0.5,4.5) (i) {};
\draw[MidnightBlue] (1.2,5.5) circle [radius=0.16];
\draw[MidnightBlue] (i)--(1.1,5.4); 

\draw[gray] (0,3.5) -- (i) -- (1,3.5);
\end{scope}

\begin{scope}[xshift=11.2cm, yshift=-3.8cm]

\node[right] at (-3.25, 3.25) {b.2)};

\draw[arrows = ->] (-3.5,2.55) -- (-3.3,2.35);

\draw[thin] (-3.15, -0.15) rectangle (-0.95,2.9);
\node[fill,circle, inner sep=2pt, color=black, label=left:$l$] at (-1.4,2) (j) {};
\draw[gray] (j) -- (-1.4, 2.75);

\draw[MidnightBlue] (-2.4,2) circle [radius=0.16];
\node[fill,circle, inner sep=2pt, color=MidnightBlue] at (-2.4,1) (i) {};
\draw[MidnightBlue] (i) -- (-2.4,1.85);

\draw[gray] (-2.9,0) -- (i) -- (-1.9,0);

\draw[arrows = ->] (-0.7,1.35) -- (-0.4,1.35);

\draw[thin] (-0.15,-0.15) rectangle (2.05,2.9);
\node[fill,circle, inner sep=2pt, color=black] at (1.1,2) (j) {};
\draw[gray] (j) -- (1.1, 2.75);

\draw[MidnightBlue] (1.1,2) circle [radius=0.16];
\node[fill,circle, inner sep=2pt, color=MidnightBlue] at (1.1,1) (i) {};
\draw[MidnightBlue] (i) -- (1.1,1.85);

\draw[gray] (0.6,0) -- (i) -- (1.6,0);

\end{scope}

\end{tikzpicture}



\end{center}
 \caption{{\bf The Wilson Balding operator.} The operator proposes a  sampled ancestor tree topology and node ages and may propose a tree of larger or smaller dimension (the number of nodes in the tree) than the original tree.  First, it prunes a subtree rooted at edge $e_1$ (blue edge) either from a branch, coloured black, in case a.1 or from a node, coloured black, in case a.2. Then it attaches the subtree either to an edge $e_2$ (black edge) at a random height in case b.1 or to a leaf $l$ (black node) in
case b.2. Case a.1 followed by b.2 removes a node from the tree and case a.2 followed by b.1 introduces a new node into the tree.} \label{fig:Pruning a subtree and reattaching it to an edge or a leaf}
 \label{fig: WBmove} \end{figure}
 
 \begin{figure}[H]
 \begin{center}
\includegraphics[width=6in]{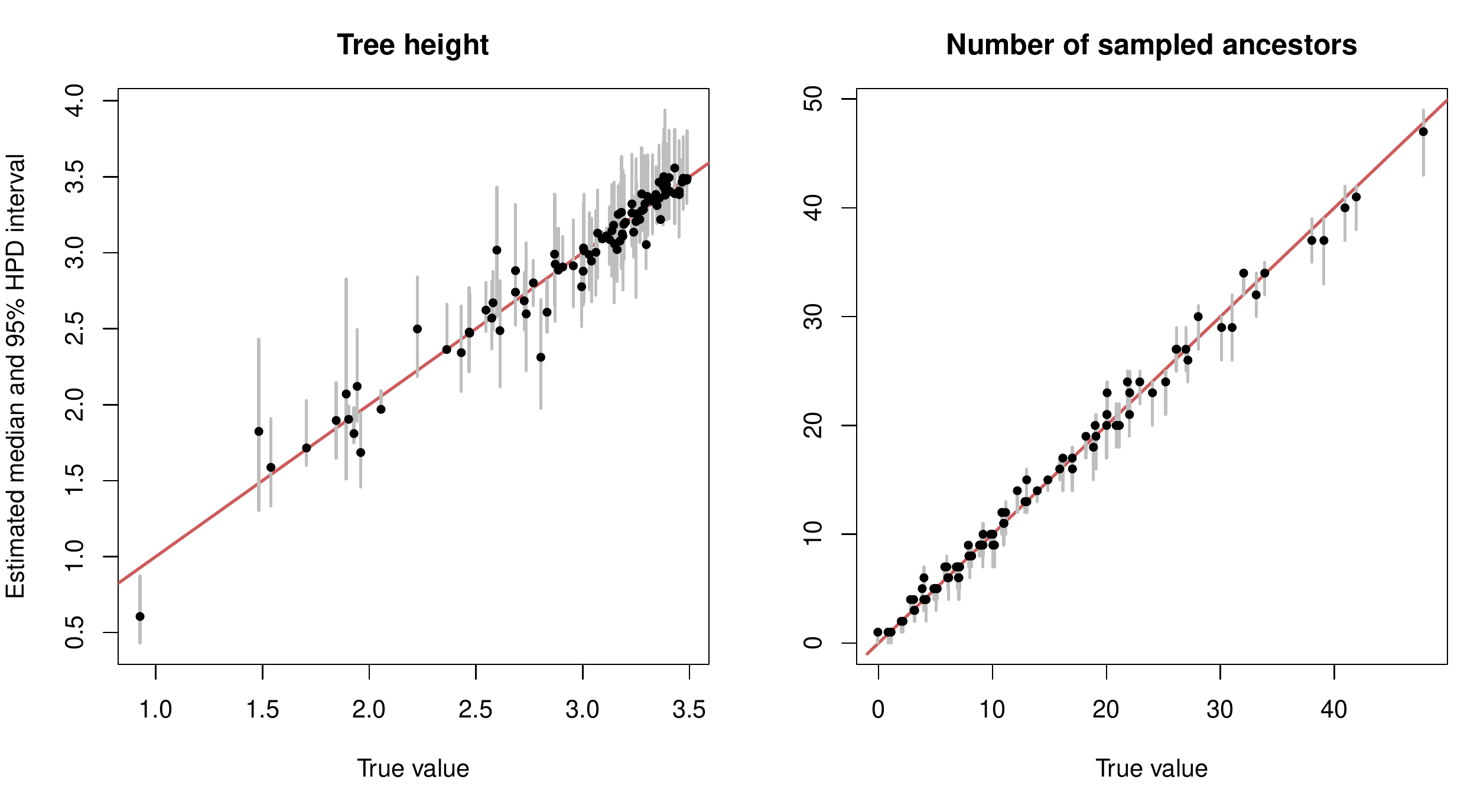}
 \end{center}
 \caption{{\bf Properties of the tree estimated from simulated data (fossilized birth-death process).} The graph shows median estimates (black dots) and 95\% HPD intervals (grey lines) against true values for the tree height (on the left) and number of sampled ancestors (on the right).}
 \label{fig: TreeChar}
 \end{figure}

\begin{figure}[H]
\begin{center} 
\includegraphics[width=3.27in]{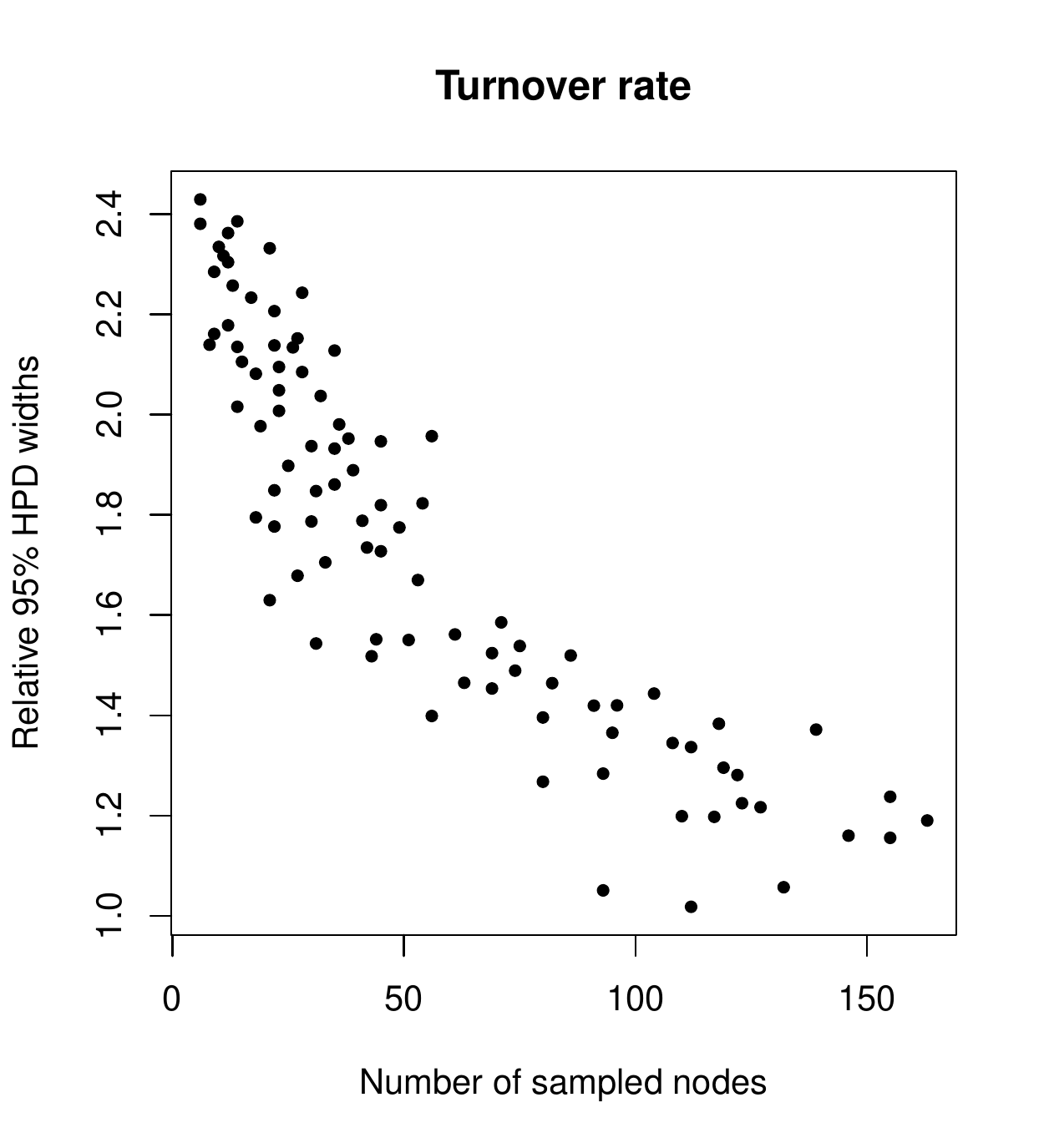}
\end{center} \caption{{\bf Uncertainty in estimates for simulated data (fossilized birth-death process).} The graph shows the widths of relative 95\%
HPD intervals of the turnover rate, $\turnover$, against tree sizes for simulated fossilized birth-death process.} 
\label{fig: TurnoverHPD} \end{figure}

 \begin{figure}[H]
 \begin{center}
 \includegraphics[width=6in]{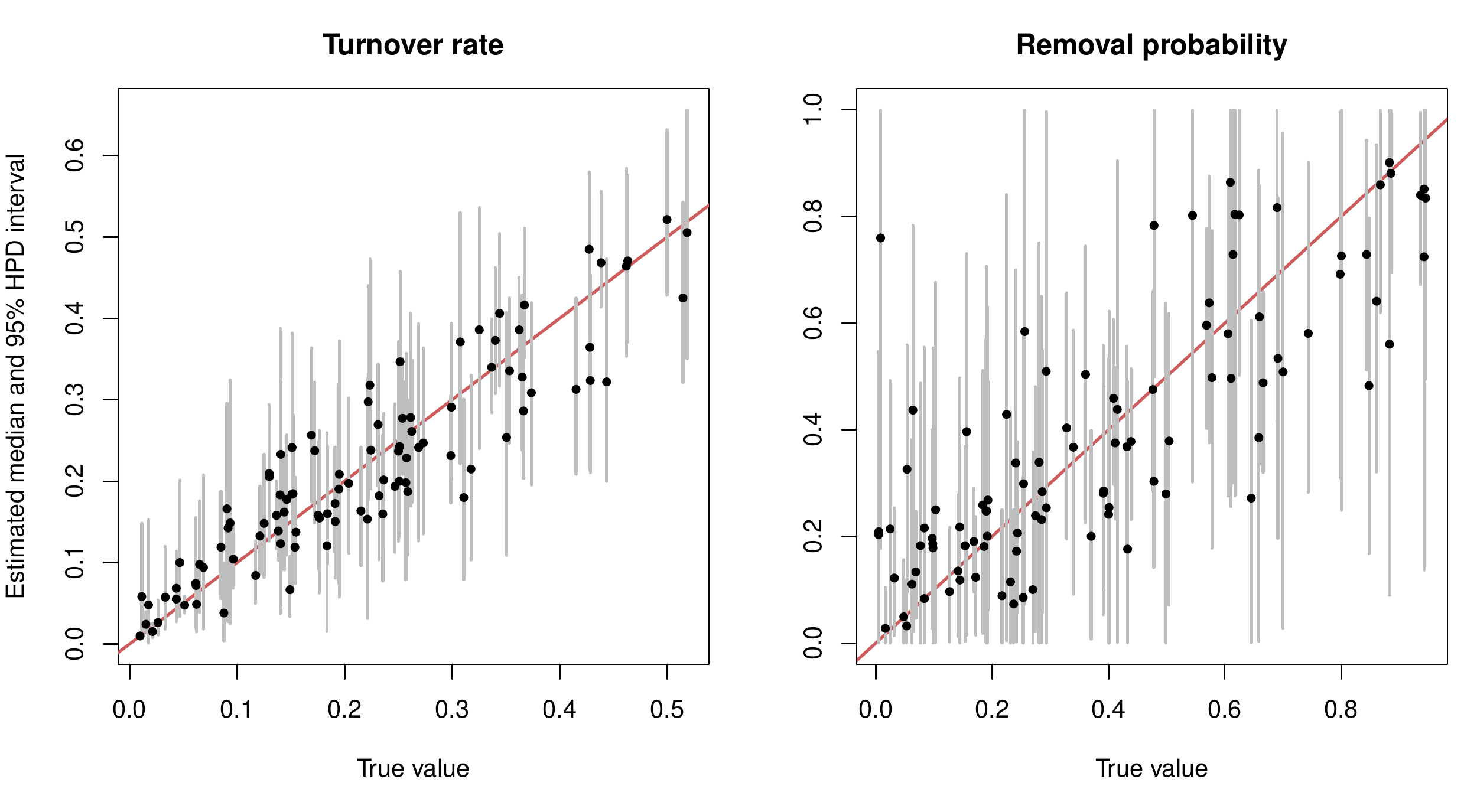}
 \end{center}
 \caption{{\bf Parameter estimates for simulated data (transmission process).} The graph shows median estimates (black dots) and 95\% HPD intervals (grey lines) against true values for the turnover rate, $\turnover$, (on the left) and removal probability, $r$, (on the right).}
 \label{fig: Parameters}
 \end{figure}

\begin{figure}[H]
 \begin{center} 
\includegraphics[width=3.27in]{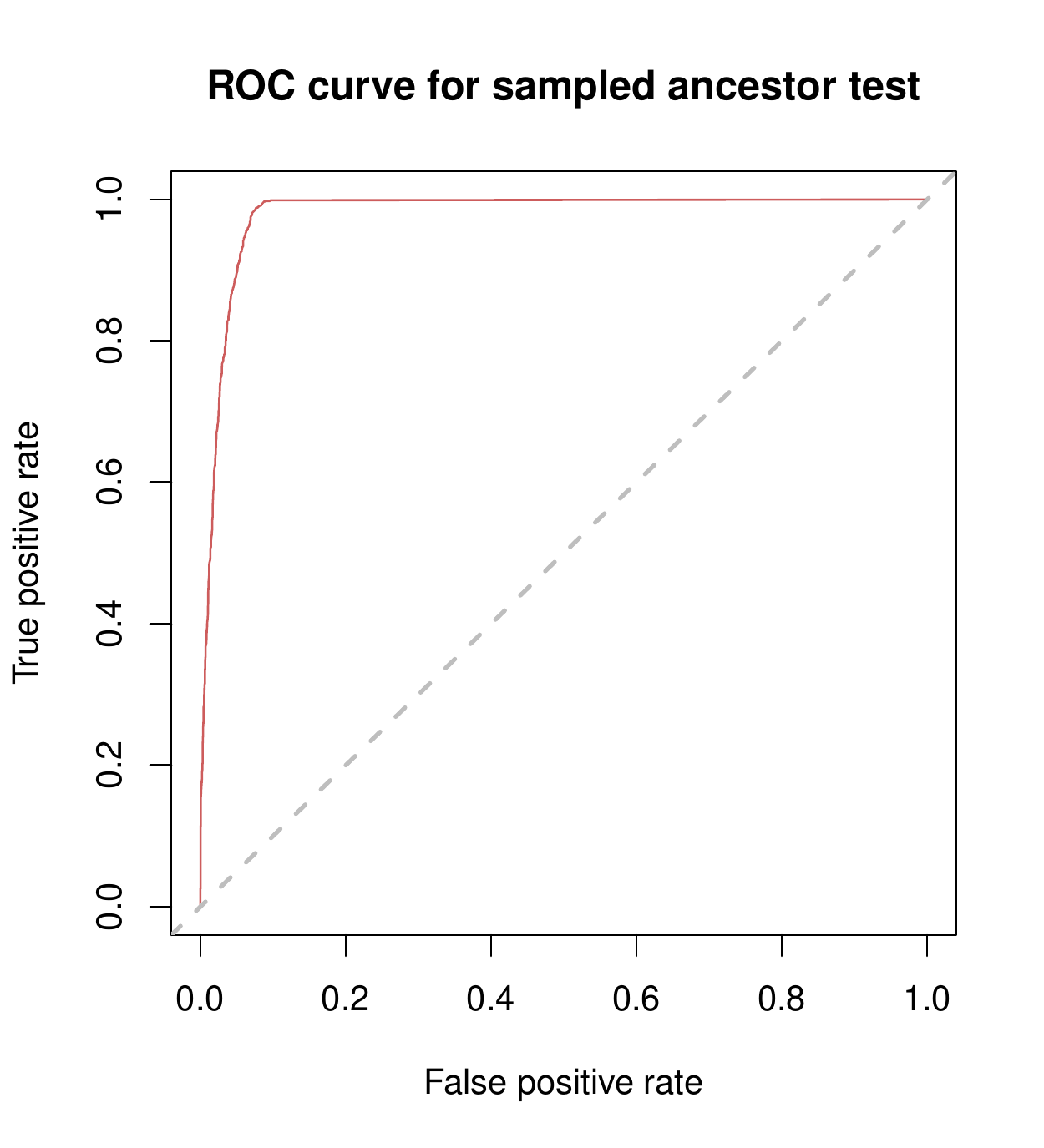}
\end{center}\caption{ {\bf ROC curve for identifying sampled ancestors based on simulated data (transmission process).}  
We identify a node as being a sampled ancestor if the posterior probability that 
the node is a sampled ancestor is greater than some threshold. The curve is parameterised 
by the threshold and shows the trade-off between true positive rate (sensitivity) and false 
positive rate (specificity) (any increase in sensitivity will be accompanied by a 
decrease in specificity) for different values of the threshold. 
The dashed diagonal line corresponds to a `random guess' test. 
The closer the ROC curve to the upper-left boarder of the ROC space (the whole area of 
the graph), the more accurate the test. The optimal value of the threshold for this curve is 0.45.}
\label{fig: Roccurve} \end{figure}

\begin{figure}[H]
 \begin{center} 
\includegraphics[width=4in]{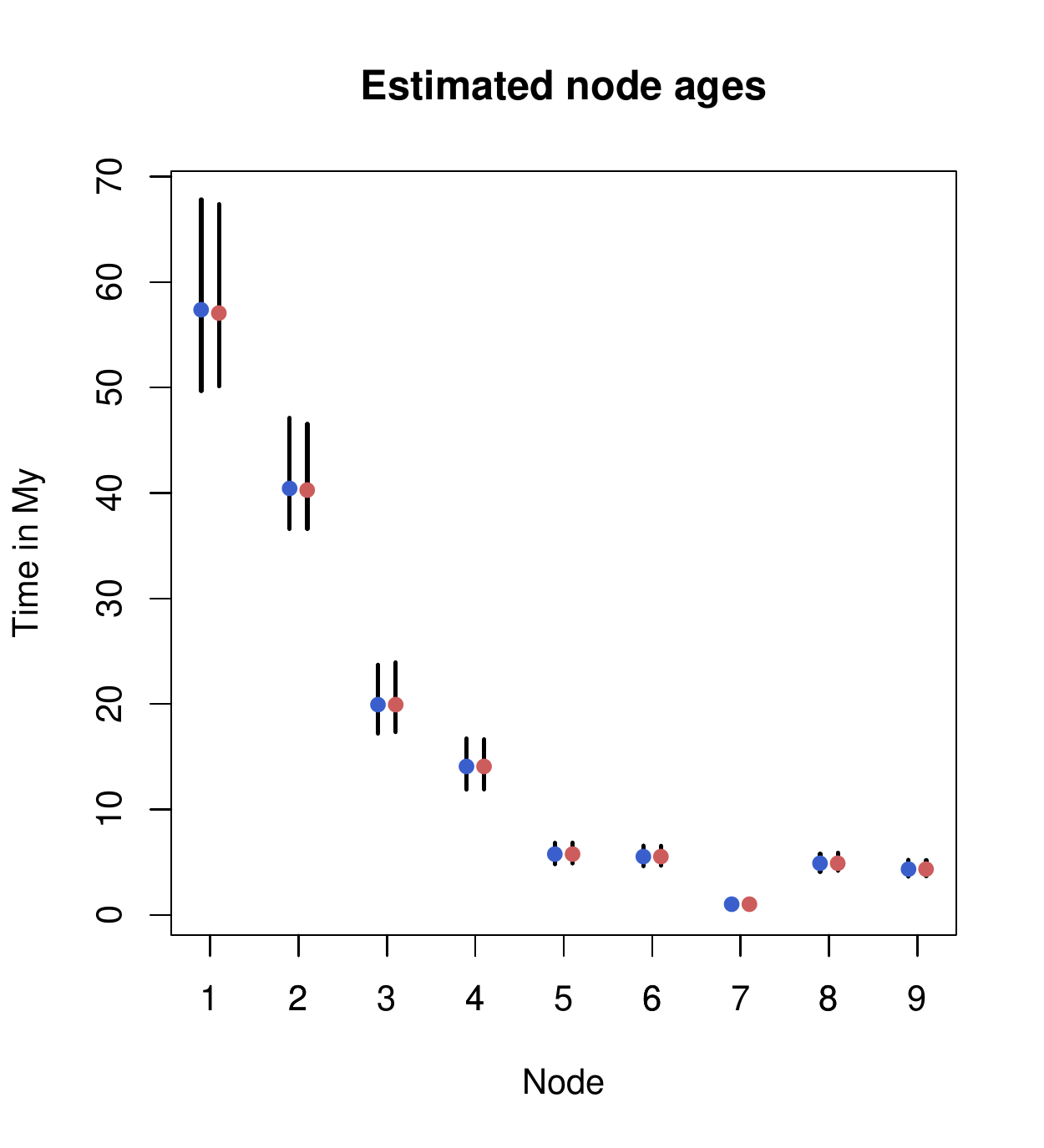}
\end{center}  \caption{ {\bf Divergence time estimates for the bear dataset.} 
The estimates are obtained from the analyses with  DPPDiv~\cite{Heath2013} 
(left bars with blue dots) and BEAST2 (right bars with red dots) implementations 
of the fossilised birth-death model, which give the same results. The bars 
 are 95\% HPD intervals and the dots are mean estimates. 
The node numbering follows the original analysis~\cite{Heath2013}: nodes 1 
and 2 represent the most recent common ancestors of
the bear clade and two outgroups (gray wolf and spotted seal). Node 3 is the
most recent common ancestor of all living bear species and nodes 4-9 are the
divergence times within the bear clade.} \label{fig: Divergence} \end{figure} 

\begin{figure}[H]
 \begin{center} 
\includegraphics[width=6in]{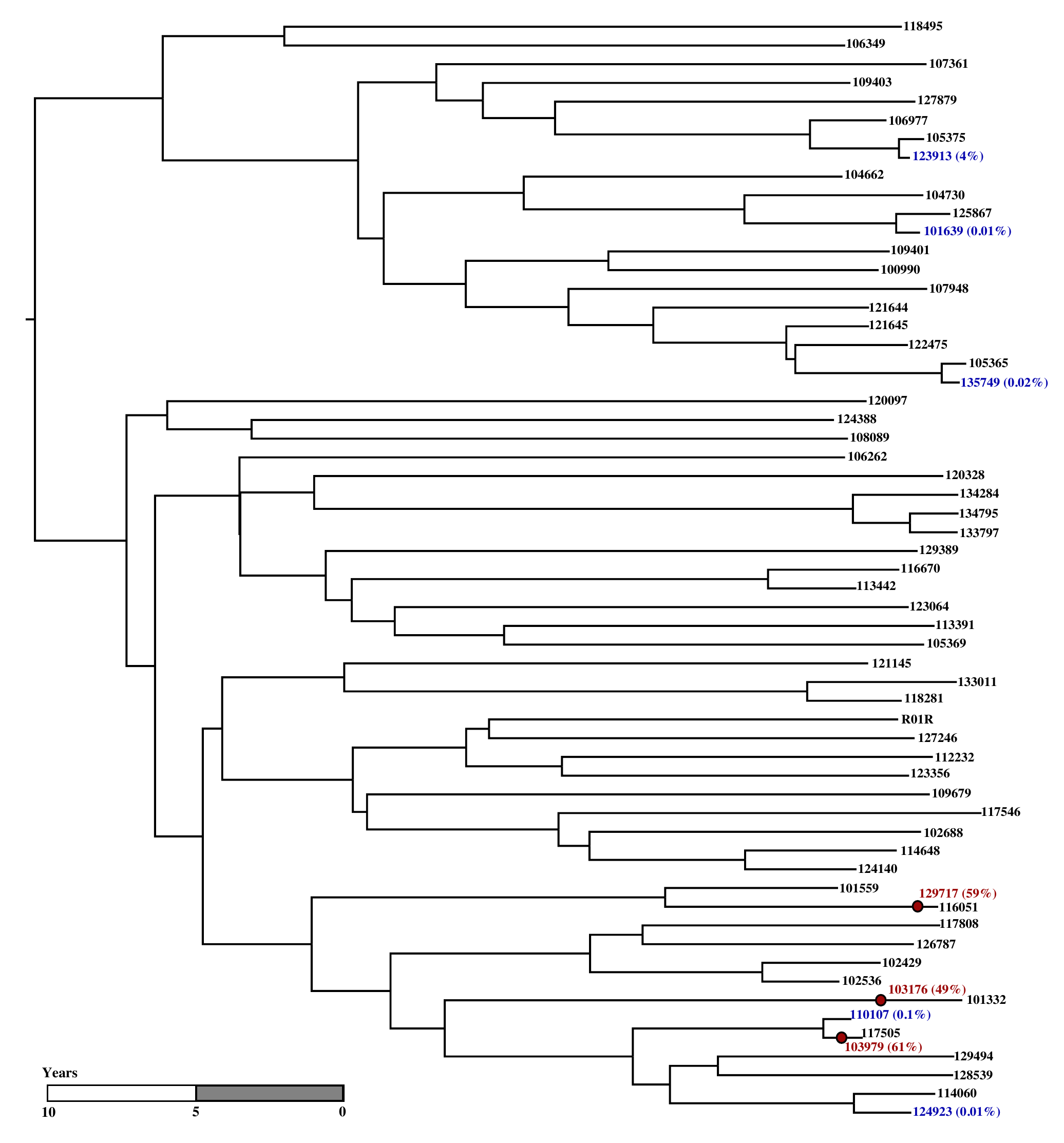}
\end{center} \caption{{\bf A tree sampled from the posterior of the HIV 1 dataset analysis. } The tree exhibits 
three estimated sampled ancestors shown as red circles. The samples with positive posterior
probabilities of being sampled ancestors are shown in colour (red for the nodes with evidence of 
being sampled ancestors and blue for other nodes with non-zero probabilities) with the posterior 
probabilities in round brackets. } \label{fig: RandomTree} \end{figure}

\section*{Tables}
%
%
\begin{table}[H] 
\caption{ \bf{Hastings ratio for the extension of the Wilson Balding operator}} 
\begin{tabular}{m{4.3cm}<{\centering}|m{2cm}<{\centering}m{2cm}<{\centering}m{2cm}<{\centering} @{}m{0pt}@{}} 
Pruning from \slash Attaching to& internal branch &  leaf  & root branch &\\ 
\hline 
internal branch &  $\frac {|I_2|}{|I_1|}$  &  $\frac  D {(D-1)} \frac 1 {|I_1|} $&  $\frac {e^{\lambda_eh_2}}{ |I_1|}$ &\\ [5ex] 
\hline 
internal node  &  $\frac D {(D+1)} |I_2|$ & $1$  &  $\frac D {(D+1)} e^{\lambda_eh_2}$ &\\ [5ex] 
\hline 
root branch &  $\frac {|I_2|}{e^{\lambda_{e}h_1}}$  &  $\frac D{(D-1)} \frac 1{e^{\lambda_eh_1}}$ &  -  &\\ [5ex] 
\end{tabular}
\begin{flushleft}The table summarises the Hastings ratio $\frac {q(g^* |
g)}{q(g|g^*)}$ for the extended Wilson Balding operator. \end{flushleft}
\label{tab:hastings} 
\end{table}

\newpage

\begin{center}
\Large{Supporting Information}
\end{center}

\section{Sampled ancestor Skyline model}

\begin{theorem} \label{ThmSky}
The probability density function for a reconstructed tree $g = (\mathcal T, \bar x, \bar y, \bar z)$ produced by the SABD skyline process with parameters $\bar \lambda, \bar \mu, \bar \psi, \bar r, \bar \rho, \bar t$ is equal to 
\begin{equation}\label{skyline}
\begin{multlined}
f[g |  \bar \lambda, \bar \mu, \bar \psi, \bar r, \bar \rho, \bar t] = \frac{2^{m+M-1}}{(m+M+k+K)!} \times \\ q_1(t_0) \prod_{i=1}^k (1-r_{\mathbf i_{z_i}})\psi_{\mathbf i_{z_i}} \prod_{i=1}^{m + M-1}\lambda_{\mathbf i_{x_i}} q_{\mathbf i_{x_i}}(x_i) \prod_{i=1}^{m}  \frac{\psi_{\mathbf i_{y_i}}(r_{\mathbf i_{y_i}} + (1-r_{\mathbf i_{y_i}})p_{\mathbf i_{y_i}}(y_i))}{q_{\mathbf i_{y_i}}(y_i)} \times \\\prod_{i=1}^\intNum ((1-\rho_i)q_{i+1}(t_i))^{n_i}\rho_i^{N_i}((1-r_{i+1})q_{i+1}(t_i))^{K_i} (r_{i+1} + (1-r_{i+1})p_{i+1}(t_i))^{M_i},
\end{multlined}
\end{equation} where, for $i = 1, \ldots, \intNum$ and $t_i \le t < t_{i-1}$,
$$p_i(t) = \frac { \lambda_i + \mu_i + \psi_i - A_i \frac {e^{A_i(t - t_i)} (1+B_i) - (1-B_i)}{e^{A_i(t - t_i)} (1+B_i) + (1-B_i)} } {2 \lambda_i}$$ with 
$$A_i = \sqrt {(\lambda_i - \mu_i - \psi_i)^2 + 4 \lambda_i \psi_i}$$ and  
$$B_i = \frac{(1-2(1-\rho_i)p_{i+1}(t_i))\lambda_i + \mu_i + \psi_i} {A_i};$$
$p_{\intNum+1}(t_\intNum) = 1$; and, for $i = 1, \ldots, \intNum$,
$$q_i(t) = \frac {4 e^{A_i(t-t_i)}} {(e^{A_i(t-t_i)}(1+B_i) + (1-B_i))^2}.$$ 
The other notation is summarised in table~\ref{tab:notation}
\end{theorem}

\begin{table} 
\begin{center}
\caption{\bf Sampled ancestor Skyline model notation}
\begin{tabular}{c|p{0.9\textwidth}}
Notation & Description \\
\hline
$\intNum$ & the number of intervals or parameter shift times,\\ 
$t_{or}$ & $t_{or} = t_0$ the time of origin, \\
$t_i$ &  is parameter shift time or $\rho$-sampling time for $i \in \{1, \ldots, \intNum \}$ with $t_\intNum = 0$, \\
$\bar t$ & $(t_0, \ldots, t_{\intNum-1})$ is a vector of those time parameters which are necessary to define the model, \\
$m$ & the number of $\psi$-sampled tips, \\
$\bar y$ & $(y_1, \ldots, y_m)$ is a vector of times of $\psi$-sampled tips, \\
$M_i$ & the number of tips sampled at time $t_i$ for $i \in \{1, \ldots, \intNum\}$, \\
$M$ & $\sum\limits_{i=1}^\intNum M_i$, \\
$\bar x$ & $(x_1, \ldots, x_{m+M})$ is a vector of bifurcation times, \\
$k$ & the number of $\psi$-sampled nodes that have sampled descendants, \\
$\bar z$ & $(z_1, \ldots, z_k)$ is a vector of times of $\psi$-sampled nodes with sampled descendants, \\
$K_i$ & the number of nodes, with sampled descendants, sampled at time $t_i$ for $i \in \{1, \ldots, m\}$\\
$K$ & $\sum\limits_{i=1}^\intNum K_i$, \\
$N_i$ & $K_i + M_i$ the total number of nodes sampled at time $t_i$ for $i \in \{1, \ldots, \intNum\}$, \\
$n_i$ & the number of lineages presented in the tree at time $t_i$ but not sampled at this time for $i~\in~\{1, \ldots, \intNum\}$, \\
${\bf i}_x$ & an index such that $t_{{\bf i}_x} \le x < t_{{\bf i}_x - 1}$.
\end{tabular}
\label{tab:notation}
\end{center}
\end{table}

\begin{proof}
First, we consider the same process but where at each bifurcation time, we label one of the new lineages as \emph{left} and another as \emph{right} and we do not label sampled nodes. In this case, the process produces oriented trees instead of labeled trees. 

The probability $p_i(t)$ that an individual alive at time $t$ has no sampled descendants when the process is stopped (i.e., in the time interval $[t_\intNum, t]$), with $t_i \le t < t_{i-1}$ ($i=1, \ldots, \intNum$) was derived in~\cite{StadKuhn12} for the birth-death skyline model without sampled ancestor (i.e. $r=1$).

Consider an event that the individual that started the process at time $t_{or}$ has no sampled descendants in the time interval $[t_\intNum, t_{or}]$. This event does not depend on the behaviour of the process if an individual was sampled (such as the possibility to remain in the process after sampling, i.e., $r<1$) because it states that no individual was sampled. That implies $p_1(t_{or}| \lambda, \mu, \psi, r) = p_1( t_{or}| \lambda, \mu, \psi)$. Since the evolution of each lineage is independent of the evolution of other coexisting lineages under this model, we can state the same for the event that an individual alive at some time $t<t_{or}$ has no sampled descendant when the process is stopped, that is, $p_i(t| \lambda, \mu, \psi, r) = p_i( t| \lambda, \mu, \psi)$ for $i$ as  above. 


For convenience, we split every edge existing at time $t_i$ (for $i = 1, \ldots, \intNum-1$) with a two degree node dated at this time. Let $g_{i,e}(t)$ be the probability density that an infected individual in the tree at time $t$ corresponding to edge $e$ (with $t_i \le t \le t_{i-1}$) evolved between $t$ and the present as observed in the tree. 

The Master equation for $g_{i, e}(t)$ along an edge with starting time $t_s$ and ending time $t_b$ ($t_e \le t \le t_b$) is
$$\frac d {dt} g_{i,e}(t) = -(\lambda_i + \mu_i + \psi_i) g_{i, e}(t) + 2 \lambda_i p_i(t) g_{i, e}(t)$$ 
Note that $r$ does not occur in the equation and will only be introduced in the initial values of $g_{i,e}$. The solution to this equation is given in~\cite{StadKuhn12}:
$$g_{i,e}(t) = g_{i,e}(t_e) \frac{q_i(t)}{q_i(t_e)}.$$ 

Further, the initial values are, for $t_e \neq t_i$,
$$g_{i, e} (t_e) = \begin{cases} 
\lambda_i g_{i, e_1}(t_e) g_{i, e_2}(t_e) &\text{ if $e$ has two descendant edges $e_1$, $e_2$,} \\
\psi_i(r_i + (1-r_i)p_i(t_e)) &\text{ if $e$ is a leaf edge, } \\
\psi_i(1-r_i)g_{i,e_1}(t_e) &\text{ if $e$ has one descendant edge $e_1$;} \\
\end{cases}
$$and for $t_e = t_i$,
$$g_{i, e} (t_e) = \begin{cases} 
(1-\rho_i)g_{i+1, e_1}(t_e) &\parbox[t]{.38\textwidth}{if $e$ has one descendant edge $e_1$ and $e$ is not a sampled node,} \\ 
\rho_i(1-r_{i+1})g_{i+1,e_1}(t_e) &\parbox[t]{.38\textwidth}{if $e$ has one descendant edge $e_1$ and $e$ is  a sampled node,}\\
\rho_i(r_{i+1} + (1-r_{i+1})p_{i+1}(t_e)) &\text{if $e$ is a leaf edge;} \\
\end{cases}
$$
 
Then the probability density of the genealogy is  
 $$f[g |  \bar \lambda, \bar \mu, \bar \psi, \bar r, \bar \rho, \bar t] =  g_{1, e_{root}}(t_0)$$
Traversing the tree from tips to the root and finding $g_{i, e}(t_e)$ for each time $t_e$ (note that $q_i(t_i)= 1$), we derive that $g_{1, e_{root}}(t_0)$ is as in~\eqref{skyline} without the first term, which comes from the fact that we considered oriented trees instead of labeled trees. 

Indeed, the expression without the first term is the probability density function for oriented trees. Having the model parameters fixed (including $\bar t$),
it only depends on branching times, sampling times of tips, sampling times of two degree nodes, the number of sampled two-degree nodes at time $t_i$, and the number of sampled tips at time $t_i$  (i.e., on $\bar x$, $\bar y$, $\bar z$, $\bar K$, and  $\bar M$), but not on how the lineages are connected, i.e. not on the particular topology. 
The density of an oriented and labeled genealogy which has the given oriented tree embedded is the oriented tree probability divided by the $(m+M+k+K)!$ possible labelings. Ignoring the $2^{m+M-1}$ orientations establishes the theorem.
\end{proof}

\subsection*{Re-parameterisation}

Let $\rho_1, \ldots, \rho_{\intNum-1} = 0$ and consider the re-parameterisation, collapsing the original $4\intNum+1$ parameters into $4\intNum$ parameters:

\begin{equation} \label{repar}
\begin{aligned}
&d_i = \lambda_i - \mu_i - \psi_i  & \text{ for  }i = 1, \ldots, \intNum  \\ 
&f_i = \lambda_i \psi_i  & \text{ for }i = 1, \ldots, \intNum \\
& g_i = (1-r_i) \psi_i & \text{ for } i  = 1, \ldots, \intNum \\ 
& h = \rho_\intNum \lambda_\intNum & \\
& k_i = \frac {\lambda_i}{\lambda_{i+1}} & \text{ for } i = 1, \ldots, \intNum - 1
\end{aligned}
\end{equation}

We will show in the following that the tree likelihood derived in Theorem \ref{ThmSky} conditioned on $\psi$-sampling at least one individual and with $\rho_1, \ldots, \rho_{\intNum -1} = 0$ depends only on the $4\intNum$ parameters obtained from the re-parameterisation, thus in the original parameter set of size $4\intNum+1$, one parameter cannot be identified from the sampled tree.

\begin{lemma}\label{lem1}
Let $\zeta_i(t) = A_i \frac {e^{A_i(t - t_i)} (1+B_i) - (1-B_i)}{e^{A_i(t - t_i)} (1+B_i) + (1-B_i)}$. Then 
\[
\begin{aligned}
&A_i = \sqrt{d_i^2 + 4f_i} & \text{ for } i &=1,\ldots, \intNum; \\
&B_i = \frac{k_i (d_{i+1} + \zeta_{i+1}(t_i)) - d_i}{A_{i+1}} & \text{ for } i &=1,\ldots, \intNum-1;  \\
&B_\intNum = \frac {2h - d_\intNum}{A_\intNum}; &&\\
&r_i + (1-r_i)p_i(t)  =  \frac{g_i}{2 f_i} \Bigl(\frac{2 f_i}{g_i} - d_i - \zeta_i(t) \Bigr)& \text{ for } i&=1,\ldots, \intNum;  \text{ and} \\
&1- p_{i+1}(t) = \frac {d_{i+1} + \zeta_{i+1}(t)} {2 \lambda_{i+1}}  &\text{ for } i&=0,\ldots, \intNum-1.
\end{aligned}
\]
\end{lemma}

\begin{proof}
We show only one case:
$$r_i + (1-r_i)p_i(t) = r_i + (1-r_i) \frac { \lambda_i + \mu_i + \psi_i - \zeta_i(t)} {2 \lambda_i} = $$ $$ \frac{(1-r_i)}{2 \lambda_i} \Bigl(\frac{2 \lambda_i r_i}{(1-r_i)} + \lambda_i + \mu_i + \psi_i - \zeta_i(t)\Bigr) =$$ 
$$\frac{g_i}{2 f_i} \Bigl(\frac{2 \lambda_i}{(1-r_i)} - 2\lambda_i + \lambda_i + \mu_i + \psi_i - \zeta_i(t) \Bigr) = \frac{g_i}{2 f_i} \Bigl(\frac{2 f_i}{g_i} - d_i - \zeta_i(t) \Bigr)$$
Other equations can be verified by substituting parameters $\bar d$, $\bar f$, $h$, and $\bar k$ with expressions given in~\eqref{repar}. 

\end{proof}

\begin{theorem}

When $\rho_1, \ldots, \rho_{\intNum-1} = 0$, the tree density function for the sampled ancestor model conditioned on $\psi$-sampling at least one individual, which is $$f[g |\bar \lambda, \bar \mu, \bar \psi, \bar r, \rho_\intNum, \bar t, S] \propto$$ $$\rho_\intNum^{N_\intNum} \frac {q_1(t_0)}{1-p_1(t_0)}  \prod_{i=1}^k (1-r_{\mathbf i_{z_i}})\psi_{\mathbf i_{z_i}} \prod_{i=1}^{m + N_\intNum-1}2 \lambda_{\mathbf i_{x_i}} q_{\mathbf i_{x_i}}(x_i) \prod_{i=1}^{m}  \frac{\psi_{\mathbf i_{y_i}}(r_{\mathbf i_{y_i}} + (1-r_{\mathbf i_{y_i}})p_{\mathbf i_{y_i}}(y_i))}{q_{\mathbf i_{y_i}}(y_i)} \prod_{i=1}^\intNum (q_{i+1}(t_i))^{n_i}$$
can be re-parameterised with parameters given in Equations~\eqref{repar}.  
\end{theorem}

\begin{proof}
We can write this function as follows:

$$\rho_\intNum^{N_\intNum}  \lambda_1 \prod_{i=1}^{m + N_\intNum-1} \lambda_{\mathbf i_{x_i}} \prod_{i=1}^{m} \psi_{\mathbf i_{y_i}}  \times $$ $$ \frac {q_1(t_0)}{\lambda_1(1-p_1(t_0))} \prod_{i=1}^k (1-r_{\mathbf i_{z_i}})\psi_{\mathbf i_{z_i}} \prod_{i=1}^{m + N_\intNum-1}2 q_{\mathbf i_{x_i}}(x_i) \prod_{i=1}^{m}  \frac{(r_{\mathbf i_{y_i}} + (1-r_{\mathbf i_{y_i}})p_{\mathbf i_{y_i}}(y_i))}{q_{\mathbf i_{y_i}}(y_i)}  \prod_{i=1}^\intNum (q_{i+1}(t_i))^{n_i}$$

From lemma~\ref{lem1}, it follows that $q_i(t)$, $\lambda_1(1-p_1(t_0))$, $r_i + (1-r_i)p_i(t)$ depend only on parameters $\bar d$, $\bar f$, $\bar g$, $h$ and $\bar k$ and do not depend on $\bar \lambda$, $\bar \mu$, $\bar \psi$, $\bar r$, $\rho_\intNum$ individually. It remains to show that $$\rho_\intNum^{N_\intNum}  \lambda_1 \prod_{i=1}^{m + N_\intNum-1} \lambda_{\mathbf i_{x_i}} \prod_{i=1}^{m} \psi_{\mathbf i_{y_i}}$$ also depends only on $\bar d$, $\bar f$, $\bar g$, $h$ and $\bar k$. 

Note that 
\[
\begin{aligned} 
& \psi_i \lambda_j  = f_i  k_i \ldots k_i & \text{ for } j < i, \\
& \rho_\intNum \lambda_i = h  k_i \ldots k_\intNum &\text{ for } j < \intNum. 
\end{aligned}
\]
and we can decompose the last term in $m+N_\intNum$ terms in either of the two forms: $\psi_i \lambda_j$ and $\psi_\intNum \lambda_i$.  
\end{proof} 

Setting $\rho_\intNum$ to zero in re-parameterisation~\eqref{repar}, we can see that function~\eqref{4m} depends on $4\intNum -1$ parameters: $\bar d$, $\bar f$, $\bar g$ and $k_1, \ldots, k_{\intNum-1}$, out of $4\intNum$ parameters: $\bar \lambda$, $\bar \mu$, $\bar \psi$ and $\bar r$. 
Also, setting $\intNum=1$, $\lambda_1 = \lambda$, $\mu_1=\mu$, $\psi_1 = \psi$, and $r_1=r$ gives us that $p_1(t) = p_0(t)$ and  $q_1(t) = q(t)$ and that $f[g | \lambda, \mu, \psi, r, t_0, S]$ is basically the same function as in~\eqref{epid}. That means that we can  re-parameterise function~\eqref{epid} with $\lambda - \mu - \psi$, $\lambda \psi$ and $\psi(1-r)$. 

In a similar manner, we can show that when $\bar r=0$, $\rho_\intNum \neq 0$ and conditioning on sampling at least one extant individual (i.e., considering skyline fossilised birth-death process with tree probability density~\eqref{epidSky}), function $\lambda_1(1-\hat p_1(t)) = \lambda_1(1 - p_1(t|\bar \psi=0))$ depends on 
\begin{equation*} 
\begin{aligned}
&\hat d_i = \lambda_i - \mu_i & \text{ for  }i &= 1, \ldots, \intNum;  \\ 
& k_i = \frac {\lambda_i}{\lambda_{i+1}} & \text{ for } i &= 1, \ldots, \intNum - 1; \text{ and} \\
&h = \rho_\intNum \lambda_\intNum. &  &
\end{aligned}
\end{equation*}
Note that $\hat d_i = d_i - g_i$ because $r_i = 0$ implying $g_i = \psi_i$ for all $i$.  That means we can re-parameterise \eqref{epidSky} with~\eqref{repar}. But for this model, we have $3\intNum + 1$ initial parameters: $\bar \lambda$, $\bar \mu$, $\bar \psi$, and $\rho_l$, and $4 \intNum$ new parameters: $\bar d$, $\bar f$, $\bar g$, $k_1, \ldots, k_{\intNum-1}$ and $h$; implying that re-parameterisation~\eqref{repar} does not reduce the number of parameters in function~\eqref{epidSky}. 

\section{Testing Operators} 

We introduced a number of operators for a random walk in the space of sampled ancestor trees and implemented the operators as a sampled ancestor add-on to the BEAST2 software. To test the implementation we run the MCMC sampler to obtain a sample from the tree distribution defined by the sampled ancestor birth-death model~\cite{Stadler2011R0}
and compare the results with calculations made in Mathematica software.

The probability density of the tree distribution is
$$f[g | \lambda, \mu, \psi, r, t_{or}]  = \frac 1{(k+m)!} q(t_{or}) (\psi(1-r))^k \prod_{i=1}^{m-1} 2 \lambda q(x_i) \prod_{i=1}^m \frac{\psi(r + (1-r)p_0(y_i))}{q(y_i)}
$$ 
We fix sample size $n = k+m$, sampling dates $\bar y$ and all the model parameters except for the time of origin $t_{or}$ placing a uniform distribution on it. So we sample genealogies from the distribution with probability density:
\begin{equation}
\label{density}
f[g, t_{or}| \lambda, \mu, \psi, r; n, \bar y] = f[g| \lambda, \mu, \psi, r, t_{or}; n, \bar y]  f_{or}(t_{or}) 
\end{equation}
where $f_{or}(x)$ is a probability density of the origin. We set,  
\begin{center}
\begin{tabular}{cccc}\label{parameters}
$\lambda = 2$ & $\mu =1$ & $\psi = 0.5$ & $r = 0.9$ \\
\multicolumn{4}{c}{$t_{or} \sim$ Uniform(0, 1000)} \\
$n= 3$ & $y_1 = 2$ & $y_2 = 1$ & $y_3 = 0$ \\
\end{tabular}
\end{center} We run 100 MCMC analysis to test different operators that do not change sampled node's times.

To assess whether the obtained tree samples are from distribution~\eqref{density} with the given parameters we calculate the true marginal probabilities for all non-ranked tree topologies on three sequentially sampled individuals in Mathematica. There are eight different non-ranked tree topologies. Denote them $T_1, \ldots, T_8$. To fix string representations for the topologies we label the individual sampled at time $y_1$ as  $1$,  at time $y_2$ as  $2$, and at time $y_3$ as  $3$. The true probabilities are shown in the table below: 

\begin{center}
\begin{tabular}{p{3cm}<\centering p{3cm}<\centering  p{2cm}<\centering}
Non-ranked tree topology  & String representation & Probability in \% \\
\hline
$T_1$ & $((3,2),1)$ & $77.8327$\\
$T_2$ & $((3,2))1$ & $7.8642$ \\
$T_3$ & $((3)2,1)$ & $3.8657$ \\
$T_4$ & $(3,(2,1))$ & $4.3189$ \\
$T_5$ & $((3,1),2)$ & $4.3189$ \\
$T_6$ & $((3)2)1$ & $0.4135$ \\
$T_7$ & $(3, (2)1)$ & $0.6930$ \\
$T_8$ & $((3)1, 2)$ & $0.6930$ \\
\hline 
\end{tabular}
\end{center}
 
Further we compare the estimated marginal probabilities with the true marginal probabilities. We calculate the standard errors of the estimated probabilities for each tree topology and assess if the estimated value is within two standard errors of the true value. To calculate a standard error given by:

$$\frac{p_{true}(1-p_{true})} {\sqrt {ESS}}$$ we need to find the number of independent samples, i.e. the effective sample size (ESS). To find the ESS we assign an integer to each of the eight topologies to obtain a sample from $\{1, \ldots, 8\}$ instead of a tree sample and calculate the ESS for the integer sample. 

The obtained results for 100 runs are summarised in the following table:

\begin{center}
\begin{tabular}{p{0.25\textwidth}|p{0.05\textwidth}<\centering p{0.05\textwidth}<\centering p{0.05\textwidth}<\centering p{0.05\textwidth}<\centering p{0.05\textwidth}<\centering p{0.05\textwidth}<\centering p{0.05\textwidth}<\centering  p{0.05\textwidth}<\centering }

\multirow{2}{*}{Operators \footnotemark[1]} & \multicolumn{8}{c}{Accuracy for non-ranked tree topologies (in \%)} \\
\cline{2-9}
 & $T_1$ & $T_2$ & $T_3$ & $T_4$ & $T_5$ & $T_6$ & $T_7$ & $T_8$ \\
\hline
WB  & 96 & 96 & 99 & 95 & 96 & 98 & 95 & 95 \\
\hline
WB & 98 & 95 & 94 &  92 &  93 & 96 & 96 & 98 \\
\hline
WB and S & 97 & 95 & 97 & 98 & 97 & 98 & 99 & 97 \\
\hline
WB, LSJ, and Ex &  92 & 97 & 97 & 99 & 98 & 98 & 95 & 96 \\
\hline
WB, LSJ, Ex, and S & 95 &  94 & 96 & 98 &  93 &  92 & 99 & 98 \\
\hline 
WB, LSJ, Ex, U, and S &  94 &  93 & 91 & 97 &  93 &  92 &  94 & 95 \\
\hline 
\end{tabular}
\end{center}

\footnotetext[1]{we use abbreviations: WB for the extension of Wilson Balding, LSJ for Leaf--sampled-ancestor jump, Ex for Exchange (assuming a combination of the narrow and wide versions), S for Scale, and U for Uniform.}

\section{Simulation studies}

We divide the simulations in two groups. In one group of simulations (Scenario 1), we simulate trees and then estimate tree model parameters with the trees fixed in MCMC. In the second type of simulations (Scenario 2), we simulate trees and sequences along the simulated trees and then run MCMC with sequences and sampled node dates as the input data to estimate tree model parameters, trees, and molecular model parameters. 

We simulate 100 trees in all the scenarios except for the last scenario of the skyline model simulations. 
When simulating trees, we either fix a set of tree model parameters and simulate each tree under the model with the fixed parameters or draw a new set of parameters from prior distributions for each tree. Further we either simulate the process until it reaches a pre-defined number of sampled nodes or until a pre-defined time length (the time of origin) is reached. 

For models without $\rho$-sampling, we fix one of the parameters to its true value in MCMC (as not all parameters can be inferred). In all scenarios, we place a uniform prior on $[0, 1000]$ for the time of origin.

We simulate sequences of 2000 bp under the GTR model with fixed rates and frequencies: $$( \eta_{AC}, \eta_{AG}, \eta_{AT}, \eta_{CG}, \eta_{CT}, \eta_{GT}) =  (0.4, 1.0, 0.1, 0.15, 1.04, 0.15)$$  $$(\pi_A, \pi_C, \pi_G, \pi_T) = (0.25, 0.25, 0.25, 0.25)$$ We use the strict molecular clock model with a fixed substitution rate.

For each estimated parameter, we take the median of its posterior distribution as a point estimate. We calculate the error and relative bias of the median estimate and relative 95\% high probability density (HPD) interval width. And we assess whether the true value is inside the 95\% HPD interval.

$$error = \frac {|true~value - median|} {true~value}$$

$$relative~bias =  \frac {true~value - median} {true~value}$$

$$relative~95\%HPD~width =  \frac {upper- lower} {true~value}$$
where $upper$ and $lower$ are the upper and lower bounds of the 95\% HPD interval.   

Then we summarise the statistics from 100 runs and report medians of 100 errors, 100 relative biases, and 100 relative 95\% HPD width. We also report the 95\% HPD accuracy, which is the number of times when the true value was inside the 95\% HPD interval. 

\subsection*{Simulation of the sampled ancestor birth-death model}

\subsubsection*{Scenario 1}
In scenarios 1.1, 1.2, and 1.4, we simulate under the model without $\rho$-sampling, i.e., $\rho=0$. In Scenarios 1.1 and 1.2, parameters are fixed and we stop simulations when a sample of 200 is reached. In Scenario 1.3, we simulate trees on a fixed time interval of $t_{or}$. We discard trees with too small or too large numbers of sampled nodes.  In Scenario 1.4, we draw parameters  $\lambda$, $\mu$, $\psi$, and $r$ from uniform prior distributions and simulate trees with 100 sampled nodes.

\begin{center}
\begin{tabular}{p{1.7cm}|p{0.9cm}<{\centering}cccp{1.3cm}<{\centering}p{1.9cm}<{\centering}p{2.2cm}<{\centering}}
& true value & prior & median & error & relative bias & {\footnotesize relative 95\% HPD width} & {\footnotesize 95\% HPD accuracy in \%} \\
\hline 
\multicolumn{8}{l}{Scenario 1.1: $\psi$ fixed in MCMC.} \\

\hline 
$\lambda$ & 0.9  & U(0,100)  & 0.9222  & 0.0735 & 0.0247 & 0.4056 & 96 \\
$\mu$ & 0.2 & U(0,100) & 0.2219 & 0.3885 & 0.1096 & 2.2460 & 94\\
$\psi$ (fixed) & 0.3 &  - & - & - & - & - & - \\
$r$ & 0.6 & U(0,1) & 0.5823 &  0.0880 & -0.0295 & 0.4087 & 93\\
\hline 

 \multicolumn{8}{l}{Scenario 1.2: $r$ fixed in MCMC.} \\

\hline 
$\lambda$ & 1.0 & U(0,100) & 1.0861 & 0.0921 & 0.0861 & 0.5020 & 96 \\
$\mu$ & 0.1 & U(0,100) & 0.1871 & 0.8705 & 0.8705 & 4.8774 & 98\\
$\psi$ & 0.4 & U(0,100) & 0.3841 & 0.0888 & -0.0399 & 0.4821 & 93 \\
$r$ (fixed) & 0.5 & - & - & - & - & - & - \\
\hline

\multicolumn{8}{l}{Scenario 1.3: $r = 0$ and simulations stop at $t_{or}$.} \\
\hline 
$t_{or}$  & 5.0 & U(0,1000) & 4.8545 & 0.0815 & -0.0291 & 0.4509 & 97 \\
$\lambda$ & 1.5 & U(0,100) & 1.6077 & 0.1112 & 0.0718 & 0.7094 & 93 \\
$\mu$ & 0.5 & U(0,100) & 0.6494 & 0.4088 & 0.2988 & 2.1971 & 94\\
$\psi$ & 0.2 & U(0,100) & 0.1884 & 0.1840 & -0.0579 & 0.8871 & 90\\
$\rho$ & 0.8 & U(0,1) & 0.7756 &  0.0916 & -0.0305 & 0.5446 & 97 \\
\hline
\multicolumn{8}{l}{Scenario 1.4: parameters drawn from priors.} \\
\hline 
$\lambda$ & - & U(1,1.5)  & 1.2899  & 0.0640 & -0.0106 & 0.3324 & 95 \\
$\mu$ & - & U(0.5, 1) & 0.6727 & 0.1546 & 0.0321 & 0.6520 & 92 \\
$\psi$ (fixed) & - & U(4,5) & - & - & - & - & - \\
$r$ & - & U(0,1) & 0.0499 &  0.4630 & -0.0853 & 2.0262 & 92\\
\hline 
\end{tabular}
\end{center}

\vskip2mm

\subsubsection*{Scenario 2}

In Scenarios 2.1.1 and 2.1.2, we use the model without $\rho$-sampling and stop simulations when a sample of 200 is reached. The tree model parameters are fixed.  In Scenarios 2.2 and 2.3, we use $d$, $r_t$, and $s$ parameterisation, i.e., we estimate and place priors on parameters
\begin{center}
\begin{tabular}{l}
$d = \lambda - \mu$ \\
$r_t = \frac \mu \lambda$\\
$\fosp = \frac \psi {\mu + \psi}$\\
\end{tabular}
\end{center} instead of $\lambda$, $\mu$ and $\psi$. 

In Scenario 2.2, we set $r=0$ and stop simulations when time $t_{or} = 3.5$ reached. The average number of sampled nodes is 50. We discard trees with less than 5 sampled nodes, and analyse 92 remaining trees. 

In Scenario 2.3, the tree model parameters drawn from the prior distributions. The stop simulation condition is when the time of origin reaches $3.0$. We discard trees with less than 5 or more than 250 sampled nodes, which constitutes in total 21\% of simulated trees. The average number of sampled nodes in remaining trees is 53. $s$ is fixed in MCMC. 

\begin{center}
\begin{tabular}{p{1.7cm}|p{0.9cm}<{\centering}cccp{1.3cm}<{\centering}p{1.9cm}<{\centering}p{2.2cm}<{\centering}}
& { true value} & prior\footnotemark[2]  & median & error & { relative bias} & {\footnotesize relative 95\% HPD width} & {\footnotesize 95\% HPD accuracy in \%} \\
\hline 
\multicolumn{8}{l}{Scenario 2.1.1:  $\mu_s = 0.02$ } \\
\hline 
$t_{or}$ & - & U(0,1000) & - & 0.0347 & 0.0066 & 0.2398 & 95 \\
{\footnotesize Tree height} & - & SABD & - & 0.0065 & -1e05 & 0.0404 & 97 \\  
\# SA & - & SABD & - & 0.0714 & 0.0000 & 0.3043 & 97 \\
$\lambda$ & 1.0 & U(0,100) & 1.0486  & 0.0825 &  0.0486 & 0.4231 & 93\\
$\mu$ & 0.2 & U(0, 100) & 0.2456 & 0.3436 & 0.2279 & 2.5473 & 95 \\
$\psi$ (fixed) & 0.4 & - & - & - & - & - & - \\
$r$ & 0.7 & U(0,1) & 0.6753 &  0.057 & -0.0353 & 0.3301 & 93 \\
\hline 
\multicolumn{8}{l}{Scenario 2.1.2:  $\mu_s = 0.0002$} \\
\hline 
$t_{or}$ & - & U(0,1000) & - & 0.0504 & 0.0009 & 0.2856 & 94 \\
{\footnotesize Tree height} & - & SABD & - & 0.0221 & -0.0067 & 0.1423 & 96\\  
\# SA & - & SABD & - & 0.3784 & 0.2660 & 1.9307 & 98 \\
$\lambda$ & 1.0 & U(0,100)  & 1.1206 & 0.1343 & 0.1206 & 0.7940 & 95\\
$\mu$ & 0.2 & U(0,100)  & 0.3397 & 0.6984 & 0.6984 & 4.0758 & 95 \\
$\psi$ (fixed) & 0.4 & - & -& - & - & - & - \\
$r$ & 0.7  & U(0,1) & 0.5726 & 0.1915 & -0.1820 & 1.0405 & 93 \\
\hline 
\end{tabular}
\end{center}

\footnotetext[2]{SABD stands for sampled ancestor birth-death model and ln$\mathcal N(\alpha, \beta)$ is a Log-normal distribution with mean $\alpha$ and standard deviation $\beta$ in the log-transformed space.}
\begin{center}
\begin{tabular}{p{1.7cm}|p{1cm}<{\centering}cccp{1.3cm}<{\centering}p{1.9cm}<{\centering}p{2.2cm}<{\centering}}
& true value & prior\footnotemark[2]  & median & error & { relative bias} & {\footnotesize relative 95\% HPD width} & {\footnotesize 95\% HPD accuracy in \%} \\
\hline 
\multicolumn{8}{l}{Scenario 2.2: $r=0$ and simulations stop at $t_{or}$} \\
\hline 
$t_{or}$ & 3.5 & U(0,1000) & 3.5776 & 0.0857 & 0.0222 & 0.5816 & 96 \\
{\footnotesize Tree height} & - & SABD & - & 0.0170 & 0.0000 & 0.1480 & 95 \\
\# SA & - & SABD &- & 0.0241 & 0.0000 & 0.1905 & 99 \\
$\mu_s$ & 0.01& {\footnotesize  ln$\mathcal N$(-4.6, 1.25)}  & 0.0099 & 0.0342 & -0.0076 & 0.2304& 95 \\
$d$ & 1.0 & U(0,1000)  & 1.0266 & 0.1872 & 0.0266 & 1.0317 & 95 \\
$r_t$ & 0.3333 & U(0,1) & 0.3343 & 0.2236 & 0.0029 & 1.7816 & 100 \\
$s$ & 0.4444 & U(0,1)  & 0.4343 & 0.1844 & -0.0229 & 1.2984 & 98 \\
$\rho$ & 0.7 & U(0,1) & 0.6854 & 0.1116 & -0.0209 & 0.8005 & 95 \\
\hline 
\multicolumn{8}{l}{Scenario 2.3: $\rho =0$, parameters drawn from priors and simulations stop at $t_{or}$} \\
\hline 
$t_{or}$ & 3.0 & U(0,1000) & 3.0084 & 0.05096 & 0.0157 & 0.4301 & 98\\
{\footnotesize Tree height} & - & SABD & - & 0.0168 & -1e-07 & 0.1322 & 94 \\  
\# SA & - & SABD & - & 0.0493\footnotemark[3] & 0.0000\footnotemark[3] & 0.3636\footnotemark[3]  & 98 \\
$\mu_s$ & 0.01 & {\footnotesize  ln$\mathcal N$(-4.6, 1.25)}  & 0.0100 & 0.0550 & 0.0023 & 0.2515 & 93 \\
$d$ & - & U(1,2)  & 1.5217 & 0.1053 & 0.0023 & 0.5239 & 93 \\
$r_t$ & - & U(0,1) & 0.1956 & 0.2084 & 0.0037 & 0.8309 & 95 \\
$r$ &- & U(0,1) & 0.3146 & 0.2814 & -0.0065 & 1.3304 & 92 \\
$s$ (fixed) & - & U(0.5, 1) & - & - & - & - & \\ 
\hline 
\end{tabular}
\end{center}

\footnotetext[3]{To calculate errors, relative biases and relative HPD widths for \#SA we increased true value, median estimate and lower and upper HPD estimates by one because the relative statistics are not defined if a true value is equal to zero.}

\subsection*{Simulation of the sampled ancestor skyline model}

In all three scenarios for simulation of the skyline model, we simulate the process until a sample of 200 is reached. We only simulate trees and do not simulate sequences in these scenarios. In Scenario 1.1, there are two intervals and  only sampling rate shifts from zero to non-zero value. In Scenario 1.2, there are two intervals and all parameters except $r$, which is  fixed in  the MCMC, shifts at time $t_1 = 5.0$. In Scenario 1.3, we have three intervals with the shift times: $t_1 = 3.0$ and $t_2= 6.0$. 
All parameters shift.  All elements of  vector $\bar r$ are fixed in MCMC.  In this last scenario, we simulated 50 trees and present the results on 42 successful MCMC runs (other runs did not converge with the chain length of 20M states).

\begin{center}
\begin{tabular}
{p{1.7cm}|p{0.9cm}<\centering cccp{1.3cm}p{1.9cm}<\centering p{2.2cm}<\centering}
\hline 
& {\small true value} & prior & median & error & {\small relative bias} & {\footnotesize  relative 95\% HPD width} & {\footnotesize  95\% HPD accuracy in \%} \\
\hline 
\multicolumn{8}{l}{Scenario 1.1: two intervals and only $\psi$ shifts from zero to non-zero value.} \\
\hline 
$\lambda$ & 0.8 & U(0,100) & 0.8107 & 0.0861 & 0.0134 & 0.4577 & 92 \\
$\mu$ & 0.4 & U(0,100) & 0.4199 & 0.2070 & 0.0499 & 1.1330 & 94 \\
$\psi$ (fixed) & 0.2 & - & - & - & - & - & -\\
$r$ & 0.8 & Uniform(0,1) & 0.7874 & 0.0394 & -0.0158 & 0.2516 & 97 \\
\hline 
\multicolumn{8}{l}{Scenario 1.2: two intervals and all parameters except $r$ shift.} \\
\hline 
$\lambda_1$ & 1.0  & U(0, 100) & 1.1869 & 0.2072 & 0.1869 & 1.1345 & 100 \\
$\lambda_2$ & 0.8 & U(0, 100) & 0.8442 & 0.065 & 0.0552 & 0.6402 & 100 \\
$\mu_1$ & 0.2 & U(0, 100) & 0.3660 & 0.8298 & 0.8298 & 6.0261 & 100 \\
$\mu_2$ & 0.2 & U(0, 100) & 0.2640 & 0.4035 & 0.3200 & 3.1496 & 100 \\
$\psi_1$ & 0.4 & U(0, 100) & 0.3452 & 0.2056 & -0.1371 & 0.9341 & 94 \\
$\psi_2$ & 0.5 & U(0, 100) & 0.4847 & 0.0915 & -0.0305 & 0.5592 & 96 \\
$r$ (fixed) & 0.7 & - & - & - & - & - & -\\
\hline
\multicolumn{8}{l}{Scenario 1.3: three intervals, all parameters shift and vector $\bar r$ fixed.} \\
\hline 
$\lambda_1$ & 1.5  & U(0, 100) & 1.6988 & 0.2385 & 0.1325 & 1.2336 & 95 \\
$\lambda_2$ & 1.2 & U(0, 100) & 1.3945 & 0.2014 & 0.1621 & 0.8813 & 95 \\
$\lambda_3$ & 0.5 & U(0, 100) & 0.5480 & 0.1568 & 0.0960 & 1.2625 & 100 \\
\specialrule{0.1pt}{0pt}{0pt}
$\mu_1$ & 0.5 & U(0, 100) & 0.7108 & 0.5132 & 0.4216 & 3.5732 & 100 \\
$\mu_2$ & 0.6 & U(0, 100) & 0.8086 & 0.4067 & 0.3477 & 1.9862 & 90 \\
$\mu_3$ & 0.2 & U(0, 100) & 0.2594 & 0.4856 & 0.2968 & 3.2805 & 100 \\
\specialrule{0.1pt}{0pt}{0pt}
$\psi_1$ & 0.4 & U(0, 100) & 0.4057 & 0.2359 & 0.0141 & 1.2262 & 90 \\
$\psi_2$ & 0.5 & U(0, 100) & 0.4497 & 0.1706 & -0.1006 & 0.6650 & 90 \\
$\psi_3$ & 0.1 & U(0, 100) & 0.0967 & 0.1933 & -0.0327 & 1.0046 & 98 \\
\specialrule{0.1pt}{0pt}{0pt}
$r_1$ (fixed) & 0.1 & - & - & - & - & - & - \\
$r_2$ (fixed) & 0.5 & - & - & - & - & - & -  \\
$r_3$ (fixed) & 0.9 & - & - & - & - & - & -  \\
\hline
\end{tabular}
\end{center} 

\section{HIV-1 data analysis}

For some of the taxon names in the tree in Figure~8M  the accession numbers are given in the table.

\begin{center}
\begin{tabular}{cc|cc}
\hline
taxon name & accession number & taxon name & accession number\\
\hline
129717 & AY362152.1 & 103979 & AY362145.1 \\ 
126787 & AY362151.1 & 134795 & AY362101.1 \\ 
102429 & AY362149.1 &134284 & AY362100.1 \\ 
103176 & AY362148.1 & 101559 & AY362057.1 \\ 
124923 & AY362147.1 & 102536 & AY362056.1 \\
 117505 & AY362146.1 & RO1R & AF494119.1 \\ 
\hline 
\end{tabular}
\end{center}

\end{document}